\documentclass[a4paper,UKenglish,cleveref,thm-restate,capitalise,noabbrev,nameinlink,pdfa]{lipics-v2021}

\usepackage{ebproof}
\usepackage{mathtools}
\usepackage{stmaryrd}
\usepackage{tikz,tikz-cd}

\title{Abstract clones for abstract syntax}

\author%
  {Nathanael Arkor}
  {University of Cambridge, UK \and \url{https://www.cl.cam.ac.uk/~na412/}}
  {na412@cl.cam.ac.uk}
  {https://orcid.org/0000-0002-4092-7930}
  {}
\author%
  {Dylan McDermott}
  {Reykjavik University, Iceland \and \url{https://dylanm.org/}}
  {dylanm@ru.is}
  {https://orcid.org/0000-0002-6705-1449}
  {Icelandic Research Fund project grant \textnumero{}~196323-053}

\authorrunning{N. Arkor and D. McDermott}
\Copyright{Nathanael Arkor and Dylan McDermott}



\nolinenumbers
\hideLIPIcs

\EventEditors{Naoki Kobayashi}
\EventNoEds{1}
\EventLongTitle{6th International Conference on Formal Structures for Computation and Deduction (FSCD 2021)}
\EventShortTitle{FSCD 2021}
\EventAcronym{FSCD}
\EventYear{2021}
\EventDate{July 17--24, 2021}
\EventLocation{Buenos Aires, Argentina (Virtual Conference)}
\EventLogo{}
\SeriesVolume{195}
\ArticleNo{30}

\newcommand{\catname}[1]{\mathbf{#1}}
\newcommand{\Clone}[1]{\catname{Clone}(#1)}
\newcommand{\Set}{\catname{Set}}
\newcommand{\Alg}{\operatorname{\catname{-Alg}}}

\newcommand{\Model}{\mathcal M}

\newcommand{\STLCTerms}{\Lambda}
\newcommand{\STLCBetaEtaTerms}{\Lambda_{\beta\eta}}
\newcommand{\STLCBetaEtaBoolTerms}{\Lambda_{\beta\eta\clonename{Bool}}}
\newcommand{\Ty}{\mathsf{Ty}}
\newcommand{\To}{\Rightarrow}

\newcommand{\syntax}[1]{\mathsf{#1}}
\newcommand{\sapp}{\syntax{app}}
\newcommand{\sget}{\syntax{get}}
\newcommand{\sput}{\syntax{put}}
\newcommand{\strue}{\syntax{true}}
\newcommand{\sfalse}{\syntax{false}}
\newcommand{\sif}[3]{\syntax{if}~#1~\syntax{then}~#2~\syntax{else}~#3}

\newcommand{\sort}[1]{\mathsf{#1}}
\newcommand{\A}{A}
\newcommand{\B}{B}
\newcommand{\basety}{\sort b}
\newcommand{\ec}{\diamond}

\newcommand{\oper}[1]{\mathsf{#1}}

\newcommand{\presname}[1]{\mathbf{#1}}
\newcommand{\SIGMA}{\presname \Sigma}
\newcommand{\tmsubst}[2]{#1 \{#2\}}
\newcommand{\metasubst}[2]{#1 \{#2\}}
\newcommand{\tmeq}{\approx}
\newcommand{\Tm}[1][]{\mathsf{Term}^{#1}}
\newcommand{\Var}[1]{\mathsf{Var}_{#1}}
\newcommand{\M}{\textsl{\textsc{m}}}
\newcommand{\incfree}[1]{#1}
\newcommand{\Bool}{\mathbb B}

\newcommand{\GSPres}[1]{\SIGMA^{\presname{GS}}_{#1}}
\newcommand{\BoolPres}{\SIGMA_{\presname{Bool}}}

\newcommand{\STLCSig}{\Sigma^{\STLCTerms}}
\newcommand{\STLCBetaEta}{\SIGMA^{\STLCBetaEtaTerms}}

\newcommand{\clonename}[1]{\mathbf{#1}}
\newcommand{\PClone}{\clonename P}
\newcommand{\X}{\clonename X}
\newcommand{\Y}{\clonename Y}
\newcommand{\VarClone}[1]{\clonename{Var}_{#1}}

\newcommand{\STLCClone}{\clonename \STLCTerms}
\newcommand{\STLCBetaEtaClone}{\clonename \STLCBetaEtaTerms}
\newcommand{\STLCBetaEtaBoolClone}{\clonename \STLCBetaEtaBoolTerms}
\newcommand{\TmClone}[1][]{\clonename{Term}^{#1}}
\newcommand{\BoolClone}[1][]{\clonename{Bool}}
\newcommand{\GSClone}[1]{\clonename{GS}_{#1}}

\newcommand{\novalueparen}[1]{\IfNoValueF{#1}{(#1)}}
\NewDocumentCommand{\var}{o}{\mathsf{var}^{\novalueparen{#1}}}
\newcommand{\substfunction}[1][]{\mathsf{subst}_{#1}}
\newcommand{\subst}[2]{#1 [#2]}
\newcommand{\lift}[1][]{\mathsf{lift}_{#1}}
\newcommand{\rename}{\triangleright}
\NewDocumentCommand{\wk}{o}{\mathsf{wk}^{\novalueparen{#1}}}
\NewDocumentCommand{\weaken}{o}{\mathsf{weaken}^{\novalueparen{#1}}}
\newcommand{\bind}[1]{{\Uparrow^{#1}}}

\newcommand{\sem}[1]{\llbracket #1 \rrbracket}

\newcommand{\compose}{\circ}

\newcommand{\extend}[1]{#1^\dagger}

\newcommand{\substpred}[1]{#1^{\sharp}}
\newcommand{\osubstpred}[1]{#1^{\flat}}

\newcommand{\ie}{i.e.}

\newcommand{\tuple}[1]{\langle#1\rangle}
\newcommand{\btt}{\mathrm{tt}}
\newcommand{\bff}{\mathrm{ff}}

\newcommand{\Ne}{\mathsf{Ne}}
\newcommand{\Nf}{\mathsf{Nf}}

\newcommand{\brho}{\boldsymbol\rho}
\newcommand{\bsigma}{\boldsymbol\sigma}

\newtheorem{notation}[theorem]{Notation}

\begin{document}

\maketitle

\begin{abstract}
  We give a formal treatment of simple type theories, such as the
  simply-typed $\lambda$-calculus, using the framework of abstract
  clones.
  Abstract clones traditionally describe first-order structures, but by
  equipping them with additional algebraic structure, one can further
  axiomatize second-order, variable-binding operators.
  This provides a syntax-independent representation of simple type
  theories.
  We describe multisorted second-order presentations, such as the
  presentation of the simply-typed $\lambda$-calculus, and their clone-theoretic algebras; free
  algebras on clones abstractly describe the syntax of simple type
  theories quotiented by equations such as $\beta$- and $\eta$-equality.
  We give a construction of free algebras and derive a corresponding induction
  principle, which facilitates syntax-independent proofs of properties such as adequacy and normalization for simple type theories.
  Working only with clones avoids some of the complexities inherent in
  presheaf-based frameworks for abstract syntax.
\end{abstract}

\section{Introduction}

The abstract concept of type theory is crucial in the study of programming languages. However, while it is generally appreciated that the concrete syntax associated to a type theory is peripheral to its fundamental structure, conventional techniques for working with type theories and proving properties thereof are predominantly syntactic. The primary reason for this incongruity is that, though abstract frameworks for defining and reasoning about general classes of type theories have been developed (e.g. \cite{fiore2010algebraic, fiore2010second, arkor2020algebraic, fiore2013multiversal, hirschowitz2010modules, ahrens2016modules, ahrens2019modular, hirschowitz2020modules}, there called \emph{second-order abstract syntax}), the mathematical prerequisites are significant and often appear unapproachable to those without a firm category theoretic background.
This is regrettable, because these general techniques alleviate much of the rote associated to syntactic proofs, such as those for adequacy, normalization, and the admissibility of substitution.

It so happens that there exists in the mathematical folklore an approach that is particularly well-suited to capturing the essential structure of simple type theories and yet requires essentially no experience with category theory to employ fruitfully: this is the formalism of \emph{abstract clones} (often simply called \emph{clones}) with algebraic structure. The structure of an abstract clone captures the notion of a context-indexed family of terms, closed under variable projection and substitution; equipping clones with algebraic structure permits the expression of variable-binding operators, like the $\lambda$-abstraction operator familiar from $\lambda$-calculi. It is known amongst cognoscenti that abstract clones might be employed for this purpose: for instance, Fiore, Plotkin, and Turi \cite{fiore1999abstract} proved that abstract clones are equivalent to their notion of \emph{substitution monoids}, which represent families of (unityped) terms with an associated capture-avoiding substitution operation; later, Fiore and Mahmoud \cite{mahmoud2011second, fiore2014functorial} proved that abstract clones with algebraic structure are equivalent to the \emph{$\Sigma$-monoids} of Fiore et al., which extend substitution monoids with second-order (\ie{} variable binding) algebraic structure. In a separate line of inquiry, Hyland \cite{hyland2012classical} uses abstract clones with algebraic structure to give a modern treatment of the unityped $\lambda$-calculus. However, it does not appear that abstract clones have previously been expressly proposed for the study of simple type theories (in fact, the definition of a typed abstract clone with algebraic structure is absent from the literature).

Here, we give an exposition of the use of abstract clones with algebraic structure in defining simple type theories and proving various of their properties.
After setting up the relevant definitions (\cref{sec:clones}), we
describe how simple type theories can be modelled by algebras of
second-order presentations (\cref{sec:second-order}).
We then show that free algebras exist, giving an abstract description of
the syntax of the type theory (\cref{sec:free-algebras}).
We derive an induction principle~\cite{lehmann1981algebraic} that enables abstract reasoning about
the syntax (\cref{sec:induction}), and show that this is powerful enough
to prove non-trivial properties of type theories, in particular using
logical relations (\cref{sec:logical-relations}).
We also compare the clone-theoretic framework to other approaches
(\cref{comparison}).
Though we do not expect our treatment to be surprising to experts
familiar with prior categorical developments, it is an important
perspective in the understanding of simple type theories and deserves explication.

Though we occasionally make reference to category theory throughout the
paper, knowledge of category theory is not necessary to understand the
content.

\section{Abstract clones and first-order presentations}
\label{sec:clones}

A typed (or \emph{multisorted}) \emph{abstract clone} \cite{taylor1993abstract}, henceforth simply \emph{clone}, encapsulates the structure of terms in simple contexts, closed under variables and substitution.
Informally, for each context $x_1 : \A_1, \ldots, x_n : \A_n$ and type $B$, where $\A_1$ to $\A_n$ are types (or \emph{sorts}), a clone $\X$ specifies a set of terms $X(\A_1, \ldots, \A_n; B)$, each element of which is considered a term of type $B$ in the context $x_1 : \A_1, \ldots, x_n : \A_n$.
It also specifies terms $\var_i$ representing
the projection of the variable $x_i$ from the context, and
functions
$\substfunction[\Gamma; A_n, \dots, A_n; B] : X(\A_1, \dots, \A_n; B) \times X(\Gamma; \A_1) \times \cdots \times X(\Gamma; \A_n) \to X(\Gamma; B)$
representing simultaneous substitution:
\begin{align*}
    t \in X(\A_1, \ldots, \A_n; B) \quad &\textit{represents} \quad x_1 : \A_1, \ldots, x_n : \A_n \vdash t : B
    \\
    \var[\A_1, \ldots, \A_n]_i \in X(\A_1, \ldots, \A_n; \A_i) \quad &\textit{represents} \quad x_1 : \A_1, \ldots, x_n : \A_n \vdash x_i : \A_i
    \\
    \substfunction[\Gamma; A_1, \dots, A_n; B](t, u_1, \ldots, u_n) \quad &\textit{represents} \quad \Gamma \vdash \tmsubst{t}{x_1 \mapsto u_1, \ldots, x_n \mapsto u_n} : B
\end{align*}
The clone $\X$ is required to satisfy laws expressing that (1)~substituting variables for themselves does nothing; (2)~applying a substitution to a variable results in the term corresponding to that variable in the substitution; and (3)~substitution is associative.

\begin{notation}
    We fix a set $S$ of types (sorts).
    We denote by $S^*$ the free monoid on $S$, \ie{} lists of elements of $S$. Conceptually, contexts $x_1 : \A_1, \ldots, x_n : \A_n$ are given by elements $[\A_1, \ldots, \A_n] \in S^*$, since variable names carry no information.
    We write $\ec \in S^*$ for the empty context, and $\Gamma, \Xi$
    for the concatenation of $\Gamma \in S^*$ and
    $\Xi \in S^*$.
    For contexts $\Gamma, \Delta \in S^*$, where $\Delta = [\A_1, \ldots, \A_n]$, we define $X(\Gamma; \Delta) = \prod_{i\leq n} X(\Gamma; \A_i)$.
    We call the elements $\bsigma \in X(\Gamma; \Delta)$ \emph{substitutions}; a substitution is therefore a tuple $\boldsymbol{\sigma} = (\sigma_1, \ldots, \sigma_n)$ of terms $\sigma_i \in X(\Gamma; \A_i)$.
\end{notation}

\begin{definition}\label{def:clone}
  An \emph{$S$-sorted clone} $\X = (X, \var, \substfunction)$ consists
  of
  \begin{itemize}
    \item for each context $\Gamma \in S^*$ and sort $\A \in S$, a set
      $X(\Gamma; \A)$ of \emph{terms};
    \item for each context $\Gamma \in S^*$, a tuple
      $\var[\Gamma] \in X(\Gamma; \Gamma)$ of \emph{variables};
    \item for each pair of contexts $\Gamma, \Delta \in S^*$ and sort $\A \in S$,
      a \emph{substitution function}
      $\textstyle
        \substfunction[\Gamma; \Delta; \A]
          : X (\Delta; \A) \times X (\Gamma; \Delta)
          \to X (\Gamma; \A)
      $, which we write as
      $\subst t {\boldsymbol\sigma} = \substfunction[\Gamma; \Delta; \A](t,\boldsymbol\sigma)$;
  \end{itemize}
  such that
  \begin{align}
    \subst{\var[\A_1, \dots, A_n]_i}{\boldsymbol\sigma} &= \sigma_i
      &&\text{for each $\boldsymbol\sigma \in X(\Gamma; A_1, \dots, A_n)$ and $i \leq n$}
    \\
    \subst{t}{\var[\Gamma]} &= t
      &&\text{for each $t \in X (\Gamma; \A)$}
    \\
    \subst{t}{\subst{\sigma'_1}{\boldsymbol\sigma},\dots,\subst{\sigma'_m}{\boldsymbol\sigma}}
      &= \subst{(\subst{t}{\boldsymbol\sigma'})}{\boldsymbol\sigma}
      &&\text{for each $t \in X (\Xi; \A)$,
      $\boldsymbol\sigma' \in X (\Delta;\Xi)$,
      $\boldsymbol\sigma \in X (\Gamma; \Delta)$}
  \end{align}
  A \emph{clone homomorphism} $f : \X \to \X'$ consists of a function
  $f_{\Gamma; \B} : X (\Gamma; \B) \to X'(\Gamma; \B)$ for each
  context $\Gamma \in S^*$ and sort $\B \in S$, such that the following
  hold, where $\Delta = [\A_1, \ldots, \A_n] \in S^*$:
  \begin{align*}
    f_{\Delta; \A_i} (\var[\Delta]_i)
      &= {\var[\Delta]_i}'
      &&\text{for each $i \leq n$}
    \\
    f_{\Gamma; \B} (\subst t {\boldsymbol\sigma})
      &= \subst{(f_{\Delta; \B}(t))}
        {f_{\Gamma; \A_1}(\sigma_1), \dots,
         f_{\Gamma; \A_n}(\sigma_n)}\mathrlap{'}
      &&\text{for each $t \in X (\Delta; \B)$,
      $\boldsymbol\sigma \in X (\Gamma; \Delta)\!$}
  \end{align*}
  We write $\Clone S$ for the category of $S$-sorted clones and
  homomorphisms.
\end{definition}

We extend every clone homomorphism $f : \X \to \X'$ to act on
substitutions as follows, where $\Delta = [\A_1, \ldots, \A_n] \in S^*$:
\begin{align*}
    f_{\Gamma; \Delta} & : X(\Gamma; \Delta) \to X'(\Gamma; \Delta) &
    f_{\Gamma; \Delta}&(\boldsymbol\sigma)
    = (f_{\Gamma; \A_1}(\sigma_1), \dots,
       f_{\Gamma; \A_n}(\sigma_n))
\end{align*}

\begin{example}
  We denote by $\VarClone S$ the $S$-sorted \emph{clone of variables}, whose family of terms is given by
  $\Var S (\A_1, \ldots, \A_n; \B) = \{i \mid \A_i = \B\}$;
  whose variables are given by $\var[\Gamma]_i = i$; and whose substitution is given by
  $i[\boldsymbol\sigma] = \sigma_i$.
  $\VarClone S$ is the initial object in $\Clone S$: for any $S$-sorted clone $\X$,
  there is a unique homomorphism $\rename : \VarClone S \to \X$ given by
  $\rename_{\Gamma; \B}(i) = \var[\Gamma]_i$.
\end{example}

\begin{example}\label{eg-monoid}
  The terms of any universal algebra \cite{birkhoff1935structure} form a \emph{monosorted} clone (\ie{} an $S$-sorted clone for which $S$ is a singleton $\{ * \}$). The sets of terms, along with the variables and substitution function, exactly match the classical notions. For instance, monoids form a clone $\mathbf{Mon}$, where $\mathrm{Mon}(\underbrace{*, \ldots, *}_n; *)$ is the free monoid on $n$ elements.
\end{example}

\begin{example}\label{example:stlcclone}
  Let $\Ty$ be the set of sorts freely generated by a base type
  $\basety \in \Ty$ and function types $(A \To B) \in \Ty$ for $A, B \in \Ty$ (precisely, $\Ty$ is the free magma on $\{ \basety \}$).
  The terms of the simply typed $\lambda$-calculus (STLC) form a
  $\Ty$-sorted clone $\STLCClone$.
  Consider terms generated by the following rules:
  \[
    \begin{prooftree}[center=false]
      \infer0{\Gamma, x : A, \Delta \vdash x : A}
    \end{prooftree}
    \qquad
    \begin{prooftree}[center=false]
      \hypo{\Gamma \vdash f : A \To B}
      \hypo{\Gamma \vdash a : A}
      \infer2{\Gamma \vdash \sapp\,f\,a : B}
    \end{prooftree}
    \qquad
    \begin{prooftree}[center=false]
      \hypo{\Gamma, x : A \vdash t : B}
      \infer1{\Gamma \vdash \lambda x : A.\,t : A \To B}
    \end{prooftree}
  \]
  (We write $\sapp$ to distinguish application of $\lambda$-terms from
  application of mathematical functions.
  We also use named variables for readability, identifying
  $\alpha$-equivalent terms.)
  Capture-avoiding simultaneous substitution
  $\tmsubst{t}{x_i \mapsto u_i}_i$ of terms is
  defined in the usual way by recursion on $t$:
  \begin{gather*}
    \tmsubst{x_j}{x_i \mapsto u_i}_i = u_j
    \qquad
    \tmsubst{(\sapp\,f\,a)}{x_i \mapsto u_i}_i
      = \sapp\,(\tmsubst{f}{x_i \mapsto u_i}_i)
          \,(\tmsubst{a}{x_i \mapsto u_i}_i)
    \\
    \tmsubst{(\lambda x : A.\,t)}{x_i \mapsto u_i}_i
      = \lambda y : A.\,(\tmsubst{t}{x_1\mapsto u_1, \dots, x_n\mapsto u_n, x\mapsto y})
  \end{gather*}
  The clone $\STLCClone$ has sets of terms
  $
    \STLCTerms (A_1, \dots, A_n; B)
      = \{ x_1 : A_1, \dots, x_n : A_n \vdash t : B \}
  $, variables $\var[\Gamma]_i = x_i$, and substitution $\subst t
  {\boldsymbol\sigma} = \tmsubst{t}{x_i \mapsto \sigma_i}_i$.

  There is a related $\Ty$-sorted clone $\STLCBetaEtaClone$ of STLC terms up to
  $\beta\eta$-equality,
  defined by quotienting the sets of terms associated to $\STLCClone$ by the equivalence relation $\tmeq_{\beta\eta}$, where
  $\Gamma \vdash t \tmeq_{\beta\eta} t' : A$ is the congruence relation generated by the following rules:
  \[
    \begin{prooftree}
      \hypo{\Gamma, x : A \vdash t : B}
      \hypo{\Gamma \vdash u : A}
      \infer2[$(\beta)$]{\Gamma \vdash
        \sapp\,(\lambda x : A.\,t)\,u
        \tmeq_{\beta\eta} \tmsubst{t}{x \mapsto u} : B}
    \end{prooftree}
    \qquad
    \begin{prooftree}
      \hypo{\Gamma \vdash t : A \To B}
      \infer1[$(\eta)$]{\Gamma \vdash
        t \tmeq_{\beta\eta} \lambda x : A.\,\sapp\,t\,x : A\To B}
    \end{prooftree}
  \]
\end{example}

\begin{remark}
    \label{remark-simple-type-theories}
    We shall only consider abstract clones with \emph{sets} of types. However, as illustrated by the previous example, the types in a simple type theory often have algebraic structure themselves.
    By considering only the underlying set of types, the algebraic structure is forgotten. This simplifies the development, at the cost of some loss of expressivity. By specifying a (monosorted) clone of types, rather than a set, one recovers exactly the \emph{simple type theories} of Arkor and Fiore \cite{arkor2020algebraic}.
\end{remark}

$\Clone{S}$ is a cartesian category, permitting us to combine clones pointwise.
The terminal object $\clonename 1$ is the unique clone in which every
set of terms is a singleton.
The binary product
$\X_1 \times \X_2$
has sets of terms given by products of sets
$
(X_1 \times X_2)(\Gamma; \A)
  = X_1(\Gamma; \A) \times X_2(\Gamma; \A)
$,
variables $\var[\Gamma]_i = (\var[\Gamma]_i, \var[\Gamma]_i)$, and
substitution
$
(t_1, t_2)[(\sigma_{11}, \sigma_{21}), \dots, (\sigma_{1n}, \sigma_{2n})]
  = (t_1[\boldsymbol\sigma_1]
    ,t_2[\boldsymbol\sigma_2])
$.

\begin{remark}
  \label{clones-form-a-variety}
  $S$-sorted abstract clones form a \emph{variety} in the sense of universal algebra; this means that $\Clone S$ is the category of models for a (multisorted) algebraic theory. Such categories are well-behaved, and several of the properties we mention throughout the paper (such as being cartesian) follow abstractly from this observation. We often choose to be more explicit for ease of comprehension, but make note where this abstract perspective is helpful.
\end{remark}

\subsection{Substitution and context extension}
\label{subst-and-renaming}
We briefly consider the structure of substitutions $\boldsymbol\sigma$ in
$S$-sorted clones $\X$, in particular to define various substitutions that
we use below, and to characterize context extension in clones.
If $\boldsymbol\sigma \in X(\Gamma; \Delta)$ and
$\boldsymbol\sigma' \in X(\Delta; \Xi)$ are substitutions, then their
\emph{composition} $(\boldsymbol\sigma' \boldsymbol\compose \boldsymbol\sigma) \in X(\Gamma; \Xi)$
is the substitution $(\sigma'_1[\boldsymbol\sigma], \dots, \sigma'_m[\boldsymbol\sigma])$,
where $m$ is the length of $\Xi$.
The three equations in the definition of a clone (\Cref{def:clone}) equivalently state
(1 \& 2) that $\var$ is the (left- and right-) unit for composition
($
  \var[\Delta] \compose \boldsymbol\sigma
  ~=~\boldsymbol\sigma
  ~=~\boldsymbol\sigma \compose \var[\Gamma]
$);
and (3) that composition is associative
($
  \boldsymbol\sigma'' \compose (\boldsymbol\sigma' \compose \boldsymbol\sigma)
    ~=~(\boldsymbol\sigma'' \compose \boldsymbol\sigma') \compose \boldsymbol\sigma
$).
In fact, this perspective underlies the connection between abstract clones and cartesian multicategories (which may be considered categories whose morphisms have multiple inputs, corresponding to each of the variables in a context): we elaborate on this connection in \Cref{comparison}.

We call the substitutions $\brho \in \Var S(\Gamma; \Delta)$
\emph{variable renamings}. This is justified by observing that $\brho$ selects a variable in the
context $\Delta$ for each variable in $\Gamma$.
If $t \in X(\Delta; \A)$ is a term in some clone $\X$, then
$\subst{t}{\rename \brho} \in X(\Gamma; \A)$ corresponds to the term in which the variables in $t$ have been
renamed according to $\brho$.
A special case of renaming is \emph{weakening}
$
  \wk[\Gamma]_{\Xi}
    = (1, \dots, n) \in \Var S(\Gamma, \Xi; \Gamma)
$. Using weakening and composition, we may define the \emph{lifting} of a
substitution $\boldsymbol\sigma \in X(\Gamma; \Delta)$ to a larger context:
\[
  \lift[\Xi] (\boldsymbol\sigma)
    ~=~(\boldsymbol\sigma \compose (\rename \wk[\Gamma]_{\Xi})
      , \rename(n+1, \dots, n+m))
   ~\in~X(\Gamma, \Xi; \Delta, \Xi)
\]
where $n$ is the length of $\Gamma$ and $m$ is the length of
$\Xi$.

Context extension induces the following operation on clones.
Given an $S$-sorted clone $\X$ and context
$\Xi \in S^*$, we let
$\bind{\Xi} \X$ be the $S$-sorted clone with terms
$(\bind{\Xi} X)(\Gamma; \A) = X(\Gamma, \Xi; \A)$, variables
$(\var[\Gamma, \Xi]_{i})_{i \le n} \in X(\Gamma, \Xi; \Gamma)$,
and substitution
$
  \subst{t}{\bsigma,\rename(n+1, \dots, n+m)} \in X(\Gamma, \Xi; \A)
$
for
$t \in X(\Delta, \Xi; \A)$ and $\bsigma \in X(\Gamma, \Xi; \Delta)$,
where $n$ is the length of $\Gamma$ and $m$ is the length of $\Xi$.
This satisfies a universal property as follows.
Weakening forms a homomorphism
$\weaken[\Xi]_{\X} : \X \to \bind \Xi \X$ that sends
$t \in X(\Gamma; \A)$ to
$\subst{t}{\rename \wk[\Gamma]_{\Xi}} \in X(\Gamma, \Xi; \A)$.
Then, for every homomorphism $g : \bind \Xi \X \to \Y$, we
obtain a homomorphism $g \compose \weaken[\Xi]_{\X} : \X \to \Y$ and
a substitution $g_{\ec; \Xi} (\var[\Xi]) \in Y(\ec; \Xi)$.
Together, these uniquely determine $g$: to give a homomorphism
$g$ is just to give a homomorphism $\X \to \Y$ and a closed term
$\sigma_i$ for each extra variable from $\Xi$.
(From the perspective of algebraic theories, context extension $\bind{\Xi} \X$ corresponds to the construction of the \emph{polynomial} \cite{lambek1988higher} or \emph{simple slice category} \cite{jacobs1999categorical} over $\Xi$.)
\begin{lemma}\label{context-extension-up}
  For each clone
  homomorphism $f : \X \to \Y$ and substitution
  $\boldsymbol\sigma \in Y(\ec; \Xi)$, there is a unique homomorphism
  $g : \bind \Xi \X \to \Y$ such that
  $g \compose \weaken[\Xi]_{\X} = f$ and
  $g_{\ec; \Xi} (\var[\Xi]) = \boldsymbol\sigma$.
\end{lemma}
\begin{proof}
  Suppose $g$ is such a homomorphism.
  Then, for each $t \in X(\Gamma, \Xi; A)$, we have
  $
    g_{\Gamma; A}(t)
    =
    \subst
      {(g_{\Gamma, \Xi; A}(\weaken[\Xi]_{\X}(t)))}
      {\var[\Gamma],
        (g_{\ec; \Xi}(\var[\Xi])) \compose \rename\wk[\ec]_\Gamma}
    =
    \subst
      {(f_{\Gamma, \Xi; A}(t))}
      {\var[\Gamma], \bsigma \compose \rename\wk[\ec]_{\Gamma}}
  $,
  where the first equality uses preservation of variables and
  substitution by $g$, and the second uses the assumptions on $g$.
  Hence, $g$ is unique when it exists.
  For existence, define
  $
    g_{\Gamma; A}(t)
    =
    \subst
      {(f_{\Gamma, \Xi; A}(t))}
      {\var[\Gamma], \bsigma \compose \rename\wk[\ec]_{\Gamma}}
  $.
\end{proof}
Substitutions $\boldsymbol\sigma \in Y(\ec; \Xi)$ are in natural bijection with
homomorphisms $\bind{\Xi}{\Var S} \to \Y$, and so \Cref{context-extension-up} equivalently states that $\bind {\Xi} \X$ is the coproduct
of $\X$ and $\bind{\Xi}{\Var S}$.
(This contrasts with presheaf-based
frameworks~\cite{fiore1999abstract,hofmann1999semantical}, in which
context extension is exponentiation.)

\subsection{First-order presentations}

Clones describe collections of terms closed under variable projection and substitution.
We will frequently be interested in clones equipped with extra structure, so as, for example, to interpret the operations of a given type theory. \emph{Presentations} permit the axiomatization of clones that interpret various operations, subject to sets of axioms; while the \emph{algebras} for a given presentation are exactly those clones that satisfy the axiomatization. Later, we will see how clones may be freely generated from presentations, allowing one to define a clone simply by specifying its generating operators and axioms.

Our treatment of first-order presentations is the classical notion of presentation for multisorted universal algebra \cite{birkhoff1970heterogeneous, goguen1985completeness}.

\begin{definition}
  An \emph{$S$-sorted first-order signature} $\Sigma$ consists of a set
  $\Sigma(\Gamma; \B)$  for each $(\Gamma; \B) \in S^* \times S$.
  We call the elements $\oper o \in \Sigma (\Gamma; \B)$ the
  \emph{$(\Gamma; \B)$-ary operators}.
  \emph{Terms over $\Sigma$} are generated by the following rules:
  \[
    \begin{prooftree}[center=false]
      \infer0{\Gamma, x : \A, \Delta \vdash x : \A}
    \end{prooftree}
    \qquad
    \begin{prooftree}[center=false]
      \hypo{\Gamma \vdash t_1 : \A_1}
      \hypo{\cdots}
      \hypo{\Gamma \vdash t_n : \A_n}
      \infer3[\ \textnormal{($\oper o \in \Sigma (\A_1, \dots, \A_n; \B)$)}]{\Gamma \vdash {\oper o}(t_1, \dots, t_n) : \B}
    \end{prooftree}
  \]
  An \emph{$(\A_1, \dots, \A_n; \B)$-ary term} $t$ over $\Sigma$ is a term
  $x_1 : \A_1, \dots, x_n : \A_n \vdash t : \B$, and
  an \emph{$(\Gamma; \B)$-ary equation} over $\Sigma$ is a pair
  $(t, u)$ of $(\Gamma; \B)$-ary terms.
  An \emph{$S$-sorted first-order presentation} $\SIGMA = (\Sigma, E)$
  consists of an $S$-sorted first-order signature $\Sigma$ and, for each
  $(\Gamma; \B) \in S^* \times S$, a set $E (\Gamma; \B)$ of
  $(\Gamma; \B)$-ary equations.
\end{definition}

\begin{remark}
  Observe that the operators of a signature correspond to terms in the logic specified below (namely, first-order equational logic). In particular, a $(\Gamma; \B)$-ary operator $\oper o$, where $\Gamma = [\A_1, \ldots, \A_n] \in S^*$, may be thought of either as a function $\oper o : \A_1, \ldots, \A_n \to \B$, or as a term $x_1 : \A_1, \ldots, x_n : \A_n \vdash \oper o : B$. These perspectives are complementary, and mirror the practice in categorical logic of representing terms by morphisms.
\end{remark}

\begin{definition}\label{def:first-order-logic}
  If $\Gamma \vdash u_i : \A_i$ for $i \leq n$ and
  $x_1 : \A_1, \dots x_n : \A_n \vdash t : \B$ are terms over an
  $S$-sorted first-order signature $\Sigma$, their substitution
  $\Gamma \vdash \tmsubst{t}{x_1 \mapsto u_1, \ldots, x_n \mapsto u_n} : \B$ is defined by
  recursion on $t$ in the usual way.
  The \emph{equational logic} over an $S$-sorted first-order
  presentation $\SIGMA = (\Sigma, E)$ consists of the following rules for the congruence of $\tmeq$ under operations and substitution,
  together with reflexivity, symmetry and transitivity of $\tmeq$:
  \begin{gather*}
    \begin{prooftree}[center=false]
      \hypo{\Gamma \vdash t_1\,{\tmeq}\,u_1\,{:}\,\A_1}
      \hypo{\cdots}
      \hypo{\Gamma \vdash t_n\,{\tmeq}\,u_n\,{:}\,\A_n}
      \infer3[\ \textnormal{($\oper o \in \Sigma (\A_1, \dots, \A_n; \B)$)}]{\Gamma \vdash \oper o(t_1, \dots, t_n) \tmeq \oper o(u_1, \dots, u_n) : \B}
    \end{prooftree}
    \\[2ex]
    \begin{prooftree}[center=false]
      \hypo{\Gamma \vdash t'_1 \tmeq u_1' : \A_1}
      \hypo{\cdots}
      \hypo{\Gamma \vdash t'_n \tmeq u_n' : \A_n}
      \infer3
        [\ \textnormal{($(t, u) \in E (\A_1, \dots, \A_n; \B)$)}]
        {\Gamma \vdash \tmsubst{t}{x_i \mapsto t'_i}_i \tmeq \tmsubst{u}{x_i \mapsto u'_i}_i : \B}
    \end{prooftree}
  \end{gather*}
  The terms over $\Sigma$ form a clone
  $\TmClone[\SIGMA] = (\Tm[\SIGMA], \var, \substfunction)$,
  where $\Tm[\SIGMA] (\Gamma; \A)$ is the set of $\tmeq$-equivalence
  classes of $(\Gamma; \A)$-ary terms; the variables are
  $\var[\Gamma]_i = x_i$; and substitution is
  $\subst t {\boldsymbol\sigma} = \tmsubst{t}{x_i \mapsto \sigma_i}_i$.
  A clone $\X$ is \emph{presented by} $\SIGMA$ when
  $\TmClone[\SIGMA]$ is isomorphic to $\X$ in $\Clone S$ (that is, when there are homomorphisms $\TmClone[\SIGMA] \rightleftarrows \X$ that are mutually inverse).
\end{definition}

\begin{remark}
  A clone may have many different presentations: for instance, the clone $\mathbf{Mon}$ of monoids (\Cref{eg-monoid}) may be presented by a unit and a binary multiplication operation, or by an $n$-ary multiplication operation for each $n \in \mathbb N$ (subject to suitable axioms).
\end{remark}

\begin{example}\label{example:global-state}
  Fix a finite set $V = \{v_1, \dots, v_k\}$ of \emph{values}.
  The $\Ty$-sorted presentation $\GSPres V$ of \emph{global $V$-valued state}
  has a $\smash{(\underbrace{\basety,\dots,\basety}_k; \basety)}$-ary operator $\oper{get}$, a
  $(\basety; \basety)$-ary operator $\oper{put}_{v_i}$ for each
  $i \leq k$, and equations
  \begin{align*}
    x : \basety \vdash
      \oper{get} (\oper{put}_{v_1}(x), \dots, \oper{put}_{v_k}(x))
      &\tmeq
      x : \basety
      &&
    \\
    x_1 : \basety, \dots, x_k : \basety \vdash
      \oper{put}_{v_i} (\oper{get}(x_1, \dots, x_k))
      &\tmeq
      \oper{put}_{v_i}(x_i) : \basety
      &&\text{for each $i \leq k$}
    \\
    x : \basety \vdash
      \oper{put}_{v_i} (\oper{put}_{v_j}(x))
      &\tmeq
      \oper{put}_{v_j}(x) : \basety
      &&\text{for each $i, j \leq k$}
  \end{align*}
  Informally, the term $\oper{get}(t_1, \dots, t_n)$ gets the current
  value $v_i$ of the state and then continues as $t_i$, while the term
  $\oper{put}_{v_i}(t)$ sets the state to $v_i$ and then continues as $t$.
  (In \cref{example:stlc-initial} below, we combine this presentation
  with the STLC to obtain a call-by-name calculus with global state.
  In call-by-name calculi, effects occur at base types, so it is only necessary to axiomatize
  $\oper{get}$ and $\oper{put}_{v_i}$ operators for $\basety \in \Ty$, rather than for all types.)
  We denote by $\GSClone V$ the clone $\TmClone[\GSPres V]$ arising from
  the presentation $\GSPres V$.
\end{example}

\section{Second-order presentations}
\label{sec:second-order}

Just as first-order presentations describe algebraic structure, second-order presentations describe binding algebraic structure \cite{fiore1999abstract}. Variable-binding operators are prevalent in type theory: for instance, the $\lambda$-abstraction operator of the STLC, $\oper{let}\text{-}\oper{in}$ expressions in functional programming languages, and case-splitting in calculi with sum types.
Second-order presentations are similar to first-order presentations, except that each operator must describe its binding structure, \ie{} how many variables (and of what types) it binds in each operand.
Hence, while first-order arities have the form
$(\A_1, \dots, \A_n; \B) \in S^*\times S$, second-order arities have
the form
$
  ((\Delta_1; \A_1), \dots, (\Delta_n; \A_n); \B)
    \in (S^* \times S)^* \times S
$.
Operators of such an arity take $n$ arguments of types $\A_1, \dots \A_n$
and produce terms of type $\B$:
the length of the context $\Delta_i \in S^*$ is the number of variables
bound by the $i$\textsuperscript{th} argument; and the argument types are given by the list $\Delta_i$.
First-order operators may be expressed as second-order operators that bind no variables.

\begin{definition}
  An \emph{$S$-sorted second-order signature} \cite{fiore1999abstract, fiore2010second} consists of a set
  $\Sigma(\Psi; \B)$ for each
  $(\Psi;\B) \in (S^*\times S)^* \times S$.
  We call the elements $\oper o \in \Sigma (\Psi; \B)$ the
  \emph{$(\Psi; \B)$-ary operators}.
\end{definition}

\begin{example}
  The $\Ty$-sorted second-order signature $\STLCSig$ of the STLC
  consists of an $((\ec; A {\To} B), (\ec;A); B)$-ary operator $\oper{app}$
  and an $((A; B); (A {\To} B))$-ary operator $\oper{abs}$ for each
  $A, B \in \Ty$.
  Thus each application operator $\oper{app}$ has two arguments, neither of which bind variables;
  and each $\lambda$-abstraction operator $\oper{abs}$ has one argument, which binds one variable.
\end{example}

Just as the axioms of first-order presentations are expressed in first-order equational logic,
the axioms of second-order presentations are expressed in the second-order equational logic of Fiore and Hur \cite{fiore2010second}. Second-order equational logic extends the first-order setting with \emph{metavariables} \cite{aczel1978general, hamana2004free, fiore2008second}, which conceptually stand for parameterized placeholders for terms. Each variable $x : \A$ in first-order logic has an associated type $\A \in S$; correspondingly, each metavariable $\M : (\A_1, \ldots, \A_n; \A)$ has an associated context and type (called \emph{second-order arities} in \cite{arkor2020algebraic}). $\M$ may be thought of as a variable parameterized by $n$ terms of types $\A_1$ through $\A_n$; a nullary ($n = 0$) metavariable behaves like an ordinary variable. There are several alternative ways to describe second-order equational logic \cite{arkor2020higher}, but we follow Fiore and Hur \cite{fiore2010second} in associating to each term both a variable context and a metavariable context: a metavariable context $\Psi$ is a list of context--sort pairs $(\Delta; \A) \in S^* \times S$.
The judgment $\Psi~|~\Delta \vdash t : \A$ expresses that the term $t$ has sort $\A$ in
variable context $\Delta$ and metavariable context $\Psi$.
Below, we write $\vec x$ for a list $x_1, \dots, x_n$ of variables, $\vec x.\,t$ to
indicate binding of the variables $\vec x$ in $t$, and write
$\vec x : \Delta$ as an abbreviation of
$x_1 : \A_1, \dots, x_n : \A_n$ for $\Delta = [\A_1, \dots, \A_n]$.

\begin{definition}
  Suppose $S$ is a set and $\Sigma$ is an $S$-sorted second-order
  signature.
  \emph{Terms over $\Sigma$} are generated by the following rules for
  variables, metavariables, and operators:
  \begin{gather*}
    \begin{prooftree}
      \infer0{\Psi~|~\Gamma, x : \A, \Delta \vdash x : \A}
    \end{prooftree}
    \\
    \begin{prooftree}
      \hypo{\Psi, \M : (\A_1, \dots, \A_n; \B), \Phi~|~\Gamma \vdash t_1 : A_1}
      \hypo{\mspace{-12mu}\cdots\mspace{-12mu}}
      \hypo{\Psi, \M : (\A_1, \dots, \A_n; \B), \Phi~|~\Gamma \vdash t_n : A_n}
      \infer3{\Psi, \M : (\A_1, \dots, \A_n; \B), \Phi~|~\Gamma \vdash \M(t_1, \dots, t_n) : \B}
    \end{prooftree}
    \\
    \begin{prooftree}
      \hypo{\Psi~|~\Gamma, \vec x_1 : \Delta_1 \vdash t_1 : \A_1}
      \hypo{\mspace{-16mu}\cdots\mspace{-16mu}}
      \hypo{\Psi~|~\Gamma, \vec x_n : \Delta_n \vdash t_n : \A_n}
      \infer3[\textnormal{($\oper o\in\Sigma((\Delta_1;\A_1),\dots,(\Delta_n;\A_n);\B)$)}]
      {\Psi~|~\Gamma \vdash \oper o ((\vec x_1.\,t_1)
        , \dots, (\vec x_n.\,t_n)) : \B}
    \end{prooftree}
  \end{gather*}
  A \emph{$((\Delta_1; \A_1), \dots, (\Delta_n; \A_n); \B)$-ary term}
  over $\Sigma$ is a term
  $\M_1 : (\Delta_1, \A_1), \dots, \M_n : (\Delta_n; \A_n)~|~\ec$ $\vdash t : \B$, and a
  \emph{$(\Psi; \B)$-ary equation} is a pair $(t, u)$ of
  $(\Psi; \B)$-ary terms.
  An \emph{$S$-sorted second-order presentation} $\SIGMA = (\Sigma, E)$
  consists of an $S$-sorted second-order signature $\Sigma$ and, for
  each $(\Psi; \B) \in (S^* \times S)^* \times S$, a set
  $E (\Psi; \B)$ of $(\Psi; \B)$-ary equations over $\Sigma$.
\end{definition}

Multisorted second-order presentations may essentially be taken as a definition of \emph{simple type theory} (modulo the subtlety regarding type operators described in \Cref{remark-simple-type-theories}): just as the informal notion of \emph{algebra} was formalized through the framework of universal algebra \cite{birkhoff1935structure}, so second-order presentations facilitate a precise, formal definition of simple type theory \cite{arkor2020algebraic}.

\begin{example}\label{example:stlcpres}
  The operators of the signature $\STLCSig$ of the STLC present the following rules:
  \[
    \begin{prooftree}
      \hypo{\Psi~|~\Gamma \vdash f : A \To B}
      \hypo{\Psi~|~\Gamma \vdash a : A}
      \infer2{\Psi~|~\Gamma \vdash \oper{app}(f, a) : B}
    \end{prooftree}
    \qquad
    \begin{prooftree}
      \hypo{\Psi~|~\Gamma, x : A \vdash t : B}
      \infer1{\Psi~|~\Gamma \vdash \oper{abs}(x.\,t) : A \To B}
    \end{prooftree}
  \]
  We can then give, for each $A, B \in \Ty$, an
  $((A; B), (\ec; A); B)$-ary equation for
  $\beta$-equality, and an $((\ec; A{\To}B); (A{\To}B))$-ary equation for
  $\eta$-equality:
  \[
    \begin{array}{rr@{~}c@{~}ll@{\qquad\quad~}r}
      \M_1 : (A; B), \M_2 : (\ec; A)~|~\ec \vdash&
        \oper{app}(\oper{abs}(x.\,\M_1(x)), \M_2())
        &\tmeq& \M_1(\M_2()) &: B
        &(\beta)
      \\
      \M : (\ec; A {\To} B)~|~\ec \vdash&
        \oper{abs}(x.\,\oper{app}(\M(), x))
        &\tmeq& \M() &: A \To B
        &(\eta)
    \end{array}
  \]
  The signature $\STLCSig$ together with these equations forms the
  $\Ty$-sorted second-order presentation $\STLCBetaEta$ of
  the STLC with $\beta\eta$-equality.
  Note that second-order equations permit the expression of \emph{axiom schemata}, as axioms containing metavariables (in both the traditional and precise sense of the term ``metavariable'') \cite{fiore2013multiversal, arkor2020algebraic}. Without second-order equations, one would have to add $\beta$ and $\eta$ equations for each instantiation of the metavariables in the rules above.
\end{example}

\begin{definition}
  If $(\Psi~|~\Gamma \vdash u_i : \A_i)_i$ and
  $\Psi~|~x_1 : \A_1, \dots, x_n : \A_n \vdash t : \B$
  are terms over an $S$-sorted second-order signature $\Sigma$, then
  their \emph{substitution}
  $\Psi~|~\Gamma \vdash \tmsubst{t}{x_i \mapsto u_i}_i : \B$ is defined
  by recursion on $t$:
  \begin{gather*}
    \tmsubst{x_j}{x_i \mapsto u_i}_i = u_j
    \qquad
    \tmsubst{\M(t_1, \dots t_m)}{x_i \mapsto u_i}_i
      = \M(\tmsubst{t_1}{x_i \mapsto u_i}_i, \dots,
           \tmsubst{t_m}{x_i \mapsto u_i}_i)
    \\
    \tmsubst{\oper o((\vec{y_1}.\,t_1), \dots,
      \oper o(\vec{y_k}.\,t_k))}
      {x_i \mapsto u_i}_i
    =
    \oper o((\vec{y_1}.\,\tmsubst{t_1}{x_i \mapsto u_i}_i), \dots,
      (\vec{y_k}.\,\tmsubst{t_k}{x_i \mapsto u_i}_i))
  \end{gather*}
  (On the right-hand side of the definition on operators, the
  terms $t_i$ are weakened, and we omit from the substitution variables
  that are mapped to themselves.)
  If instead we have terms
  $
    (\Psi~|~\Gamma, \vec x_i : \Delta_i \vdash u_i : \A_i)_i
  $
  and
  $
    \M_1 : (\Delta_1; \A_1), \dots, \M_n : (\Delta_n; \A_n)~|~\Gamma'
      \vdash t : \B
  $
  then their \emph{metasubstitution}
  $
    \Psi~|~\Gamma, \Gamma' \vdash
      \metasubst{t}{\M_i \mapsto (\vec {x_i}.\,u_i)}_i : \B
  $
  is defined using ordinary substitution by recursion on $t$:
  \begin{gather*}
    \metasubst{x}{\M_i \mapsto (\vec {x_i}.\,u_i)}_i = x
    \qquad
    \metasubst{\M_j(t_1, \dots, t_m)}
      {\M_i {\mapsto} (\vec {x_i}.\,u_i)}_i
    = \metasubst{u_j}{x_{jk} \mapsto
        \metasubst{t_k}{\M_i {\mapsto} (\vec {x_i}.\,u_i)}_i}_k
    \\
    \metasubst{\oper o((\vec{y_1}.\,t_1), \dots,
      (\vec{y_k}.\,t_k))}
      {\M_i \mapsto (\vec {x_i}.\,u_i)}_i
    \\\qquad
    = \oper o(
      (\vec{y_1}.\,\metasubst{t_1}{\M_i \mapsto (\vec {x_i}.\,u_i)}_i),
      \dots,
      (\vec{y_k}.\,\metasubst{t_k}{\M_i \mapsto (\vec {x_i}.\,u_i)}_i))
  \end{gather*}
\end{definition}

\subsection{Algebras}

The algebras for a presentation are the abstract clones interpreting each of the operations of the signature, subject to the axioms of the presentation. In other words, a presentation is a specification of structure, while the algebras are the realizations, or models, of that structure. For instance, in the first-order setting, the algebras for the presentation of monoids
form (set-theoretic) monoids.

\begin{definition}
  An \emph{algebra} $(\X, \sem{{-}})$ for an $S$-sorted second-order signature
  $\Sigma$ (called ``presentation clones'' in \cite{mahmoud2011second}) consists of an $S$-sorted clone $\X$ and, for each
  context $\Gamma$ and $((\Delta_1; \A_1), \dots, (\Delta_n; \A_n); \B)$-ary
  operator $\oper o$, a function
  $
    \sem{\oper o}_\Gamma
      : \prod_i X (\Gamma, \Delta_i; \A_i) \to X (\Gamma; \B)
  $
  such that, for all substitutions $\boldsymbol\sigma \in X (\Xi; \Gamma)$ and
  tuples of terms $(t_i \in X (\Gamma, \Delta_i; \A_i))_i$,
  \[
    \subst{(\sem{\oper o}_{\Gamma}(t_1, \dots, t_n))}{\boldsymbol\sigma}
      = \sem{\oper o}_{\Xi}(
          \subst{t_1}{\lift[\Delta_1] \boldsymbol\sigma}
          , \dots,
          \subst{t_n}{\lift[\Delta_n] \boldsymbol\sigma})
  \]
  A \emph{homomorphism} $f : (\X, \sem{{-}}) \to (\X', \sem{{-}}')$ of
  $\Sigma$-algebras is a homomorphism $f : \X \to \X'$ of clones such that,
  for all
  $\oper o \in \Sigma((\Delta_1; \A_1), \dots, (\Delta_n; \A_n); \B)$
  and $(t_i\in X(\Gamma,\Delta_i;\A_i))_i$,
  \[\textstyle
    f_{\Gamma; \B}(\sem{\oper o}_{\Gamma} (t_1, \dots, t_n))
    ~=~
    \sem{\oper o}'_{\Gamma}
      (f_{\Gamma,\Delta_1; \A_1} t_1,
      \dots,
      f_{\Gamma,\Delta_n; \A_n} t_n)
  \]
\end{definition}

The interpretation of operators in a $\Sigma$-algebra
$(\X, \sem{{-}})$ extends to an interpretation
$
  \sem{t}_{\Gamma} : \prod_i X (\Gamma, \Delta_i; \A_i)
    \to X (\Gamma, \Xi ; \B)
$
of each term
$
  \M_1 : (\Delta_1; \A_1), \dots, \M_n : (\Delta_n; \A_n)
  ~|~\vec x : \Xi \vdash t : \B
$
as follows (where $n$ is the length of $\Gamma$):
\begin{align*}
  \sem{x_i}_{\Gamma}(\boldsymbol\sigma) &= \var[\Gamma, \Xi]_{n + i}
  \\
  \sem{\M_i(t_1, \dots, t_m)}_{\Gamma}(\boldsymbol\sigma) &= \subst{\sigma_i}
    {\var[\Gamma, \Xi]_1
    ,\dots
    ,\var[\Gamma, \Xi]_n
    ,\sem{t_1}_{\Gamma}(\boldsymbol\sigma)
    ,\dots
    ,\sem{t_m}_{\Gamma}(\boldsymbol\sigma)}
  \\
  \sem{\oper o((\vec x_1.\,t_1), \dots, (\vec x_m.\,t_m))}_{\Gamma}(\boldsymbol\sigma)
    &= \sem{\oper o}_{\Gamma, \Xi}
      (\sem{t_1}_{\Gamma}(\boldsymbol\sigma), \dots, \sem{t_m}_{\Gamma}(\boldsymbol\sigma))
\end{align*}

\begin{definition}
  \label{algebra-for-sop}
  An algebra $(\X, \sem{{-}})$ for a second-order
  presentation $\SIGMA = (\Sigma, E)$ is a $\Sigma$-algebra such that,
  for all equations $(t, u) \in E(\Psi; \A)$ and contexts $\Gamma$, we
  have $\sem{t}_{\Gamma} = \sem{u}_{\Gamma}$.
  We let $\SIGMA\Alg$ be the category of $\SIGMA$-algebras and
  all $\Sigma$-algebra homomorphisms between them.
\end{definition}

\begin{example}\label{example:stlc-algebras}
  An algebra for the presentation $\STLCBetaEta$ of the STLC with
  $\beta\eta$-equality consists of a $\Ty$-sorted clone $\X$ and
  functions
  \[
    \sem{\oper{app}}_{\Gamma}
      : X(\Gamma; A \To B) \times X(\Gamma; A)
        \to X(\Gamma; B)
    \qquad
    \sem{\oper{abs}}_{\Gamma}
      : X(\Gamma, A; B)
        \to X(\Gamma; A \To B)
  \]
  that commute with substitution and satisfy
  \begin{align*}
    \sem{\oper{app}}_{\Gamma} (\sem{\oper{abs}}_{\Gamma} (t),\,t')
    &=
    \subst{t}{\var[\Gamma], t'}
    &&\text{for}~t \in X (\Gamma, A; B), t' \in X (\Gamma; A)
      \tag{$\beta$}\label{eq:algebra-beta}
    \\
    \sem{\oper{abs}}_{\Gamma}
      (\sem{\oper{app}}_{\Gamma, A} (\subst{t}{\rename \wk[\Gamma]_{A}}, \var[\Gamma, A]_{n + 1}))
    &=
    t
    &&\text{for}~t \in X (\Gamma; A{\To}B)
      \tag{$\eta$}
  \end{align*}
  For each set $Z$ we have a set-theoretic interpretation of the STLC,
  which forms a $\STLCBetaEta$-algebra $(\Model_Z, \Model_Z\sem{{-}})$
  as follows.
  Define interpretations $\Model_Z\sem{A} \in \Set$ of each sort
  $A \in \Ty$ recursively by setting $\Model_Z\sem{\basety} = Z$ and
  $\Model_Z\sem{A \To B} = \Set(\Model_Z\sem{A}, \Model_Z\sem{B})$
  (where $\Set(Y, Y')$ is the set of functions $Y \to Y'$).
  We then have a $\Ty$-sorted clone $\Model_Z$, where the sets of terms
  are given by
  $
    \Model_Z(A_1, \ldots, A_n; B)
      = \Set(\prod_i \Model_Z\sem{A_i},\Model_Z\sem{B})
  $,
  the variables by projections $\var[\Gamma]_i = \pi_i$, and
  substitution by
  $f[\boldsymbol\sigma] = (\xi \mapsto f (\sigma_1(\xi), \dots, \sigma_n(\xi)))$.
  This forms a $\STLCBetaEta$-algebra, with interpretations of the
  operators given by function application and currying.
  More generally, the interpretation of the STLC in any cartesian-closed
  category $\catname C$ with a specified object $Z \in \catname C$ forms a
  $\STLCBetaEta$-algebra taking
  $
    \Model_Z(A_1, \dots, A_n; B)
      = \catname C(\prod_i \Model_Z\sem{A_i}, \Model_Z\sem{B})
  $ to be the sets of terms, where $\Model_Z\sem{\basety} = Z$ and
  $\Model_Z\sem{A \To B} = {\Model_Z\sem{B}}^{\Model_Z\sem{A}}$.
\end{example}

The cartesian structure of $\Clone S$ lifts to $\SIGMA\Alg$ for every
presentation $\SIGMA$: the clone $\clonename 1$ uniquely forms a
$\SIGMA$-algebra, and the product
$(\X_1, \sem{{-}}_1) \times (\X_2, \sem{{-}}_2)$ is the clone
$\X_1 \times \X_2$ equipped with interpretations
$
  \sem{\oper o}_{\Gamma}((\sigma_{11}, \sigma_{21}), \dots, (\sigma_{1n}, \sigma_{2n}))
    = (\sem{\oper o}_{1,\Gamma}(\boldsymbol\sigma_1)
      ,\sem{\oper o}_{2,\Gamma}(\boldsymbol\sigma_2))
$.

\section{Free algebras}
\label{sec:free-algebras}
Second-order $S$-sorted presentations $\SIGMA$ can be viewed as
descriptions of simple type theories for which $S$ is the set of types.
In particular, the operators specify the term formers of the type theory (such as
$\lambda$-abstraction, or application).
From this perspective, the syntax of the type theory described by $\SIGMA$
is the \emph{initial} $\SIGMA$-algebra: there is a unique $\Sigma$-algebra homomorphism from the algebra formed by the syntax to any other algebra, given by induction on terms.
More generally, given the syntax of an existing theory in the form of a clone $\X$, the
\emph{free} $\SIGMA$-algebra on $\X$ is given by
augmenting $\X$ by the operators and equations of $\SIGMA$; or, from another perspective, augmenting the type theory described by $\SIGMA$ with the operations specified by $\X$.
For example, the free $\STLCBetaEta$-algebra on $\GSClone V$ (\Cref{example:global-state}) may be seen as the STLC extended by additional term formers ($\oper{get}$ and $\oper{put}_{v_1}, \dots, \oper{put}_{v_k}$) representing the side-effects of global state.

\begin{definition}\label{def:free-algebra}
  Suppose $\SIGMA = (\Sigma, E)$ is an $S$-sorted second-order presentation and $\X$
  is an $S$-sorted clone.
  A $\SIGMA$-algebra $F_{\SIGMA}\X$ equipped with a clone homomorphism
  $\eta_{\X} : \X \to F_{\SIGMA} \X$ is the \emph{free $\SIGMA$-algebra}
  on $\X$ if, for any other $\SIGMA$-algebra $(\Y, \sem{{-}})$ and
  clone homomorphism $f : \X \to \Y$, there is a unique
  $\Sigma$-algebra homomorphism
  $\extend f : F_\SIGMA \X \to (\Y, \sem{{-}})$ such that
  $\extend f \compose \eta_{\X} = f$.
  The \emph{initial $\SIGMA$-algebra} is the free $\SIGMA$-algebra on
  $\Var S$.
\end{definition}

\begin{example}\label{example:stlc-initial}
  Recall the presentation $\STLCBetaEta$ of the STLC with
  $\beta\eta$-equality from \cref{example:stlcpres}.
  The initial $\STLCBetaEta$-algebra is the clone $\STLCBetaEtaClone$
  of STLC terms up to $\tmeq_{\beta\eta}$ (\cref{example:stlcclone}),
  with the operators $\oper{app}$ and $\oper{abs}$ interpreted as
  \begin{align*}
    ((f, a) \mapsto \sapp\,f\,a) &:
      \STLCBetaEtaTerms(\Gamma; A \To B)
        \times \STLCBetaEtaTerms(\Gamma; A)
        \to \STLCBetaEtaTerms(\Gamma; B)
    \\
    (t \mapsto \lambda x : A.\,t) &:
      \STLCBetaEtaTerms(\Gamma, A; B)
        \to \STLCBetaEtaTerms(\Gamma; A \To B)
  \end{align*}

  The free $\STLCBetaEta$-algebra on the clone $\GSClone V$ of global
  $V$-valued state (\cref{example:global-state}) can be described as
  follows for $V = \{v_1, \dots, v_k\}$.
  The underlying $\Ty$-sorted clone is defined in the same way as
  $\STLCBetaEtaClone$, but with the following additional term formers and equations (omitting the typing constraints on equations).
  \[
    \begin{array}{l}
    \begin{prooftree}[center=false]
      \hypo{\Gamma \vdash t_1 : \basety}
      \hypo{\cdots}
      \hypo{\Gamma \vdash t_k : \basety}
      \infer3{\Gamma \vdash \sget(t_1, \dots, t_k) : \basety}
    \end{prooftree}
    \\[1ex]
    \begin{prooftree}[center=false]
      \hypo{\Gamma \vdash t : \basety}
      \infer1[\ \textnormal{($i \leq k$)}]{\Gamma \vdash \sput_{v_i}(t) : \basety}
    \end{prooftree}
    \end{array}
    \qquad
    \begin{array}{r@{~\tmeq_{\beta\eta}~}lr}
        \sget(\sput_{v_1}(t), \dots, \sput_{v_k}(t))
        &t
        \\[0.5ex]
        \sput_{v_i} (\sget(t_1, \dots, t_k))
        &\sput_{v_i}(t_i)
        &(i \leq k)
        \\[0.5ex]
        \sput_{v_i} (\sput_{v_j}(t))
        &\sput_{v_j}(t)
        &(i,k \leq k)
    \end{array}
  \]
  This forms a $\STLCBetaEta$-algebra in the same way as
  $\STLCBetaEtaClone$ above.
  The morphism $\eta_{\GSClone V}$ is given by
  $
    \eta_{\GSClone V} (\oper{get} (t_1, \dots, t_k))
      = \syntax{get} (\eta_{\GSClone V}(t_1), \dots, \eta_{\GSClone V}(t_k))
  $
  and
  $
    \eta_{\GSClone V} (\oper{put}_{v_i}(t))
      = \syntax{put}_{v_i} (\eta_{\GSClone V}(t))
  $.
\end{example}

If $\SIGMA'$ is a first-order presentation, the free $\SIGMA$-algebra on $\TmClone[\SIGMA']$ is closed under the operators of $\SIGMA'$:
each $\oper o \in \SIGMA'(\A_1, \dots, \A_n; \B)$ induces a term
$\eta(\oper o(x_1,\dots, x_n)) \in F_{\SIGMA} \X(\A_1, \dots, \A_n; \B)$
and hence functions
$
  (\boldsymbol\sigma~\mapsto~\subst{\eta(\oper o(\vec x))}{\boldsymbol\sigma})
  \,:\,
  F_{\SIGMA} \X(\Gamma; \A_1, \dots, \A_n) \to F_{\SIGMA} \X(\Gamma; \B)
$.

We show that free algebras for any signature, and on any clone, exist, by constructing them explicitly.
Existence of these free algebras facilitates the developments in the
next sections.
However, note that we do not rely on the explicit description:
after this section, we reason about free algebras solely using the universal property in \Cref{def:free-algebra}. This is important, as we wish to reason about type theories independently of their syntax, which leads to greatly simplified proofs. (It is also possible to prove the existence of free algebras entirely
abstractly using a monadicity theorem and \Cref{clones-form-a-variety}, avoiding concrete syntax.)

\begin{figure}[t]
  \boxed{\text{Terms}}
  \begin{gather*}
    \SwapAboveDisplaySkip
    \begin{prooftree}[center=false]
      \infer0{\Gamma, x : \A, \Delta \vdash_{\X} x : \A}
    \end{prooftree}
    \qquad
    \begin{prooftree}[center=false]
      \hypo{\Gamma \vdash_{\X} t_1 : \A_1}
      \hypo{\cdots}
      \hypo{\Gamma \vdash_{\X} t_n : \A_n}
      \infer3[\ \textnormal{($f \in X(\A_1,\dots,\A_n;\B)$)}]
        {\Gamma \vdash_{\X} \incfree{f}(t_1, \dots, t_n) : \B}
    \end{prooftree}
    \\[1ex]
    \begin{prooftree}[center=false]
      \hypo{\Gamma, \vec x_1 : \Delta_1 \vdash_{\X} t_1 : \A_1}
      \hypo{\mspace{-1mu}\cdots\mspace{-1mu}}
      \hypo{\Gamma, \vec x_n : \Delta_n \vdash_{\X} t_n : \A_n}
      \infer3
        [\ \textnormal{($\oper o \in \Sigma((\Delta_1; \A_1), \dots, (\Delta_n; \A_n); \B)$)}]
        {\Gamma \vdash_{\X}
          \oper o((\vec x_1.\,t_1), \dots, (\vec x_n.\,t_n)) : \B}
    \end{prooftree}
  \end{gather*}
  \boxed{
    \text{Equations (reflexivity, symmetry, transitivity omitted)}}
  \begin{gather*}
    \begin{prooftree}[center=false]
      \hypo{\Gamma \vdash_{\X} t_1 \tmeq u_1 : \A_1}
      \hypo{\cdots}
      \hypo{\Gamma \vdash_{\X} t_n \tmeq u_n : \A_n}
      \infer3[\ \textnormal{($f \in X(\A_1, \dots, \A_n; \B)$)}]
        {\Gamma \vdash_{\X} \incfree f (t_1, \dots, t_n)
          \tmeq \incfree f (u_1, \dots, u_n) : \B}
    \end{prooftree}
    \\
    \begin{prooftree}[center=false]
      \hypo{\Gamma, \vec x_1 : \Delta_1 \vdash_{\X} t_1 \tmeq u_1 : \A_1}
      \hypo{\mspace{-1mu}\cdots\mspace{-1mu}}
      \hypo{\Gamma, \vec x_n : \Delta_n \vdash_{\X} t_n \tmeq u_n : \A_n}
      \infer3
        [\textnormal{($\oper o \in \Sigma ((\Delta_1;\A_1), \dots; \B)$)}]
        {\Gamma \vdash_{\X}
          \oper o((\vec x_1.\,t_1), \dots, (\vec x_n.\,t_n))
          \tmeq \oper o((\vec x_1.\,u_1), \dots, (\vec x_n.\,u_n)) : \B}
    \end{prooftree}
    \\[2ex]
    \begin{prooftree}[center=false]
      \hypo{\Gamma, \vec x_1{:}\Delta_1 \vdash_{\X} t_1 \tmeq u_1:\A_1}
      \hypo{\mspace{-12mu}\cdots\mspace{-12mu}}
      \hypo{\Gamma, \vec x_n{:}\Delta_n \vdash_{\X} t_n \tmeq u_n:\A_n}
      \infer3
        [\textnormal{($(t',u') \in E((\Delta_1;\A_1), \dots; \B)$)}]
        {\Gamma \vdash_{\X}
          \metasubst{t'}{\M_i \mapsto (\vec x_i.\,t_i)}_i
          \tmeq
          \metasubst{u'}{\M_i \mapsto (\vec x_i.\,u_i)}_i : \B}
    \end{prooftree}
    \\[2ex]
    \begin{prooftree}[center=false]
      \hypo{\Gamma \vdash_{\X} t_1 : \A_1}
      \hypo{\cdots}
      \hypo{\Gamma \vdash_{\X} t_n : \A_n}
      \infer3[\textnormal{($i \leq n$)}]{\Gamma \vdash_{\X}
        \incfree{t_i \tmeq \var[\A_1,\dots,\A_n]_i}(t_1,\dots,t_n) : \B}
    \end{prooftree}
    \\
    \begin{prooftree}[center=false]
      \hypo{\Gamma \vdash_{\X} t_1 : \A_1}
      \hypo{\cdots}
      \hypo{\Gamma \vdash_{\X} t_n : \A_n}
      \infer3
        {\Gamma \vdash_{\X}
          \incfree{f}
            (\incfree{\sigma_1}(t_1, \dots, t_n)
            , \dots
            , \incfree{\sigma_{k}}(t_1, \dots, t_n))
          \tmeq
          (\incfree{\subst{f}{\boldsymbol\sigma}})(t_1, \dots, t_n) : \B}
    \end{prooftree}
    \tag{$f \in X(\A'_1, \dots, \A'_k; \B),
      \boldsymbol\sigma \in X(\A_1, \dots, \A_n; \A_1', \dots, \A'_k)$}
  \end{gather*}
  \caption{%
    Construction of the free $(\Sigma, E)$-algebra on a clone
      $\X = (X, \var, \substfunction)$.}
  \label{free-algebra-construction}
\end{figure}

In universal algebra, free algebras of first-order presentations are
constructed in two steps: by first closing a sort-indexed set $X$ of constants
under the operators of the presentation; and then quotienting the terms by the equations of the
presentation.
\Cref{free-algebra-construction} gives the analogous construction in the second-order setting.
First, we construct terms $\Gamma \vdash_{\X} t : \B$ from variables, the
terms of the clone $f \in X(\A_1, \dots, \A_n; \B)$ (viewed as function symbols), and the operators of the presentation $\SIGMA$.
Second, we quotient by the equivalence relation $\tmeq$ generated by
congruence, the equations of $\SIGMA$ (using metasubstitution), and
rules imposing compatibility with the clone structure of $\X$.
The clone $F_{\SIGMA} \X$ has terms
$
  F_{\SIGMA} \X(\Gamma; \B)
    = \{\Gamma \vdash_{\X} t : \B\}/\tmeq
$, with variables and substitution defined in the evident way;
the homomorphism $\eta_{\X} : \X \to F_{\SIGMA} \X$ sends $t \in X(\Gamma; \B)$ to
$x_1 : A_1, \ldots, x_n : A_n \vdash_{\X} t(x_1, \dots, x_n) : \B$, where $\Gamma = [\A_1, \ldots, \A_n]$.

\begin{proposition}
  For every $S$-sorted second-order presentation $\SIGMA$ and $S$-sorted
  clone $\X$, the free $\SIGMA$-algebra $F_\SIGMA\X$ exists.
\end{proposition}
The forgetful functor $\SIGMA\Alg \to \Clone S$ therefore
has a left adjoint (in fact, it is monadic).

\section{Induction over second-order syntax}
\label{sec:induction}

We now describe how the formalism of abstract clones may be used to prove properties of simple type theories. To begin, we consider predicates over abstract clones, which are predicates over the terms of the type theory induced by the clone, closed under the structural operations of variable projection and substitution.
Below, we extend each family of subsets
$P(\Gamma; \A) \subseteq Y(\Gamma; \A)$
to contexts by defining $P(\Gamma; \A_1,\dots,\A_n)$ to be the set of all
substitutions $\boldsymbol\sigma \in Y(\Gamma; \A_1,\dots,\A_n)$ such that
$\sigma_i \in P(\Gamma; \A_i)$ for all $i \leq n$.

\begin{definition}
  A \emph{predicate} $P$ over an $S$-sorted clone $\X$ consists of a
  subset $P(\Gamma; \A) \subseteq X(\Gamma; \A)$ for each
  $(\Gamma; \A) \in S^* \times S$ such that, for all contexts $\Gamma = [\A_1, \ldots, \A_n]$ and $i \leq n$, we have
  $\var[\Gamma]_i \in P(\Gamma; \A_i)$, and, for all
  $t \in P(\Delta; \B)$ and $\boldsymbol\sigma \in P(\Gamma; \Delta)$, we have
  $\subst{t}{\boldsymbol\sigma} \in P(\Gamma; \B)$.
\end{definition}
Closure under variables and under substitution imply that $P$ forms a
clone $\PClone$ whose inclusion $\PClone \hookrightarrow \X$ into $\X$ is
a clone homomorphism.
Predicates over $S$-sorted clones are equivalently the subobjects in
$\Clone S$, and are hence closed under arbitrary conjunction, existential quantification, and quotients of equivalence relations.
(This follows from \Cref{clones-form-a-variety}, since varieties are exact categories \cite[Theorem~5.11]{barr1971exact}, and all exact categories enjoy these properties.)
They are also closed under context extension: if $P$ is a predicate over
$\X$ and $\Xi$ is a context, then $\bind{\Xi} P$ is a
predicate over $\bind{\Xi} \X$.

We present a meta-theorem for establishing properties of simple type theories.
\begin{theorem}[Induction principle for second-order syntax]
  \label{thm:induction}
  Suppose that $(\Y, \sem{{-}})$ is an algebra for an
  $S$-sorted second-order presentation $\SIGMA$, that $f : \X \to \Y$ is
  a clone homomorphism from an $S$-sorted clone $\X$, and that $P$ is a
  predicate over $\Y$.
  If
  \begin{itemize}
    \item for all operators
      $\oper o \in \Sigma((\Delta_1;\A_1), \dots, (\Delta_n;\A_n); \B)$,
      contexts $\Gamma \in S^*$, and tuples of terms
      $(t_i \in P(\Gamma, \Delta_i; \A_i))_i$ we have
      $\sem{\oper o}_\Gamma (t_1, \dots, t_n) \in P(\Gamma; \B)$;
    \item for all terms $t \in X(\Gamma; \A)$ we have
      $f_{\Gamma; \A}(t) \in P(\Gamma; \A)$,
  \end{itemize}
  then, for all free terms $t \in (F_{\SIGMA} \X) (\Gamma; \A)$, we have
  $\extend f_{\Gamma; \A}(t) \in P(\Gamma; \A)$.
\end{theorem}
\begin{proof}
  The predicate $P$ is closed under operators, so the interpretations of
  operators in $\Y$ make $\PClone$ into a $\SIGMA$-algebra.
  The image of $f$ is contained in $P$, so $f$ forms a clone
  homomorphism $\X \to \PClone$.
  By the universal property of the free algebra $F_{\SIGMA} \X$, we
  therefore have an algebra homomorphism $F_{\SIGMA} \X \to \PClone$.
  This necessarily sends
  $t \in (F_{\SIGMA}\X)(\Gamma; \A)$ to
  $\extend f_{\Gamma; \A}(t) \in P(\Gamma;\A)$.
\end{proof}

We give two corollaries of this induction principle.
The first is for proving properties of closed terms, which take the form
of families of subsets $P (\A) \subseteq Y(\ec; \A)$.
Given such a family $P$, let $P(\A_1,\dots, \A_n)$ be the set of all
$\boldsymbol\sigma \in Y(\ec; \A_1,\dots,\A_n)$ such that $\sigma_i \in
P(\A_i)$ for all
$i \leq n$, and define a predicate $\substpred P$ over $\Y$ by
$
  \substpred P (\Gamma; \A) = \{ t \in Y(\Gamma; \A) \mid
    \forall \boldsymbol\sigma \in P(\Gamma).\,\subst{t}{\boldsymbol\sigma} \in P (\A)
  \}
$.
Applying the induction principle above to $\substpred P$ gives us the
following.
\begin{corollary}\label{corollary:closed-terms}
  Suppose that $\SIGMA$ is an $S$-sorted second-order presentation,
  that $(\Y, \sem{{-}})$ is a $\SIGMA$-algebra, and that
  $(P (\A) \subseteq Y (\ec; \A))_{\A \in S}$ is a family of subsets.
  For every $S$-sorted clone $\X$ and clone homomorphism $f : \X \to \Y$, if
  \begin{itemize}
    \item for every operator
      $\oper o \in \Sigma ((\Delta_1;\A_1),\dots, (\Delta_n;\A_n); \B)$
      and tuple
      $
        (t_i \in \substpred P (\Delta_i;\A_i))_{i \leq n}
      $
      of terms, we have
      $\sem{\oper o}_{\ec}(t_1, \dots, t_n) \in P (\B)$;
    \item for every term $t \in X (\Gamma; \B)$, we have
      $f_{\Gamma; \B}(t) \in \substpred P (\Gamma; \B)$,
  \end{itemize}
  then, for every type $\A \in S$ and free term $t \in (F_{\SIGMA} \X) (\ec; \A)$, we
  have $\extend f_{\ec; \A}(t) \in P (\A)$.
\end{corollary}
\begin{proof}
  $\substpred{P}(\ec; \A) = P(\A)$, so it suffices to apply
  \cref{thm:induction} to the predicate $\substpred P$.
  We therefore check the two assumptions of that theorem.
  Closure of $\substpred{P}$ under $f$ is immediate; and $\substpred P$ is closed under operators because, if
  $(t_i \in \substpred P(\Gamma, \Delta_i; \A_i))_i$ and
  $\boldsymbol\sigma \in P(\Gamma)$, then
  $
    \subst{t_i}{\lift[\Delta_i] \boldsymbol\sigma}
      \in \substpred P(\Delta_i; \A_i)
  $
  for all $i \leq n$, so that
  $
    \subst{\sem{\oper o}_{\Gamma}(t_1,  \dots, t_n)}{\boldsymbol\sigma}
    = \sem{\oper o}_{\ec}(
        \subst{t_1}{\lift[\Delta_1]\boldsymbol\sigma},
        \dots,
        \subst{t_n}{\lift[\Delta_n]\boldsymbol\sigma})
    \in P(\B)
  $.
\end{proof}
Families of subsets $P(\A) \subseteq X(\ec; \A)$ are closed under
arbitrary conjunction and disjunction, complements, and universal and
existential quantification.
They form a \emph{tripos}~\cite{hyland1980tripos,pitts2002tripos}, and
hence a model of higher-order logic over $\Clone S$; the tripos-theoretic methods of Hofmann~\cite{hofmann1999semantical} carry over in this way to the setting of abstract clones.

The second corollary is for families of subsets
$P(\Gamma; \A) \subseteq Y(\Gamma; \A)$ that are not known to be closed
under substitution.
(In some cases proving closure under substitution requires an induction
over terms, but induction over terms is what this section is meant to
enable.)
Analogously to the construction $\substpred P$ for predicates over
closed terms, we
define a predicate $\osubstpred P$ over $\Y$ by
$
  \osubstpred P (\Gamma; \A) = \{ t \in Y(\Gamma; \A) \mid
    \forall \Delta, \boldsymbol\sigma \in P(\Delta; \Gamma).\,
      \subst{t}{\boldsymbol\sigma} \in P (\Delta; \A)
  \}
$.

\begin{corollary}\label{corollary:open-terms}
  Suppose that $\SIGMA$ is an $S$-sorted second-order presentation,
  that $(\Y, \sem{{-}})$ is a $\SIGMA$-algebra, and that
  $(P (\Gamma;\A) \subseteq Y (\Gamma;\A))_{(\Gamma;\A)\in S^*\times S}$
  is a family of subsets.
  For every $S$-sorted clone $\X$ and homomorphism $f : \X \to \Y$, if
  \begin{itemize}
    \item for every context $\Gamma$ we have
      $\var[\Gamma] \in P(\Gamma; \Gamma)$;
    \item for every context $\Gamma$, operator
      $\oper o \in \Sigma ((\Delta_1;\A_1),\dots,(\Delta_n;\A_n);\B)$,
      and tuple of terms
      $
        (t_i \in \osubstpred P(\Gamma, \Delta_i; \A_i))_i
      $ we have
      $
        \sem{\oper o}_\Gamma(t_1, \dots, t_n)
          \in \osubstpred P(\Gamma; \B)
      $;
    \item for every term $t \in X (\Gamma; \B)$ we have
      $f_{\Gamma; \B}(t) \in \osubstpred P (\Gamma; \B)$,
  \end{itemize}
  then, for every free term $t \in (F_{\SIGMA} \X) (\Gamma; \A)$, we
  have $\extend f_{\Gamma; \A}(t) \in P (\Gamma; \A)$.
\end{corollary}
\begin{proof}
  We can
  apply \cref{thm:induction}
  to $\osubstpred P$ because it is closed under operators and under $f$.
  Hence $\extend f_{\Gamma; \A}(t) \in \osubstpred P (\Gamma; \A)$ for
  each $t \in (F_{\SIGMA} \X) (\Gamma; \A)$, and so
  $\var[\Gamma] \in P(\Gamma; \Gamma)$ implies that
  $
    \extend f_{\Gamma; \A}(t)
      = \subst{(\extend f_{\Gamma; \A}(t))}{\var[\Gamma]}
      \in P(\Gamma; \A)
  $.
\end{proof}

The above corollaries are designed to enable logical relations
arguments, in which the fundamental lemma is proven using an induction
hypothesis that quantifies over substitutions.
In particular, in \cref{corollary:open-terms} we require $\osubstpred P$
to be closed under the operators, rather than $P$.
There is a third corollary that instead requires closure of $P$ under operators (this would essentially be the principle of induction on
$\Gamma \vdash_{\X} t : \A$), but this is less useful for our purposes.

\section{Logical relations}
\label{sec:logical-relations}
We provide two extended examples of proofs using
the induction principles of the previous section,
both involving the presentation $\STLCBetaEta$ of the STLC with
$\beta\eta$-equality.
The first is a proof of the adequacy of the set-theoretic model of
the STLC, which uses induction on closed terms; the
second is a proof that every STLC term is $\beta\eta$-equal to one in
normal form, using induction on open terms.
Both examples are logical relations proofs, the
former using ordinary
logical relations and the latter using
Kripke relations~\cite{jung1993new}.
Though both properties are known to hold, these proofs in particular illustrate that our induction principles are powerful
enough to justify logical relations arguments. We include a proof of normalization for the STLC with global state in \Cref{normalization-global-state}, as a further motivating example.

\subsection{Closed terms and adequacy}
We say that a model $\Model$ of the STLC is \emph{adequate} when, for
all closed terms $t$ and $u$ of the base type $\basety$, if
$\Model\sem{t} = \Model\sem{u}$, then $t$ and $u$ are equal up to
$\beta\eta$-equality.
(In adequate models, equality of denotations
implies observational equivalence for terms of arbitrary types.)

We first show that we can perform logical relations arguments for the
STLC using our induction principle: specifically
\cref{corollary:closed-terms}.
Fix a $\STLCBetaEta$-algebra $(\Y, \sem{{-}})$, homomorphism
$f : \X \to \Y$ from some clone $\X$, and a subset
$P(\basety) \subseteq Y(\ec; \basety)$ of closed terms of base type.
We extend $P$ to a family of subsets $P(A) \subseteq Y(\ec; A)$ in
the standard way for logical relations:
\[
  P(A \To B) = \{t \in Y(\ec; A {\To} B) \mid
    \forall a \in P(A).\,\sem{\oper{app}}_{\ec}(t, a) \in P(B)\}
\]
Applying \cref{corollary:closed-terms} to $P$ gives us the
following:
\begin{lemma}\label{lemma:lr-closed}
  If, for every context $\Gamma$ and term $t \in X(\Gamma; \B)$, we have
  $f_{\Gamma; \B}(t) \in \substpred{P}(\Gamma; \B)$, then, for every free term
  $t \in (F_{\STLCBetaEta}\X)(\ec; \A)$, we have
  $\extend f_{\ec; A}(t) \in P(\A)$.
\end{lemma}
\begin{proof}
  The only non-trivial assumption of \cref{corollary:closed-terms} is
  closure under operators.
  Closure under $\oper{app}$ is immediate from the definition of the
  logical relation.
  Closure under $\oper{abs}$ holds because, if
  $t \in \substpred P(A; B)$, then, for all $a \in P(A)$, we have
  $
    \sem{\oper{app}}_{\ec}(\sem{\oper{abs}}_{\ec}(t), a)
    =
    t[a] \in P(B)
  $ using the \ref{eq:algebra-beta} law, so that
  $\sem{\oper{abs}}_{\ec}(t) \in P(A \To B)$.
\end{proof}

Note that if terms are generated only by $\lambda$-abstraction and
application then there are no closed terms of base type.
For a more interesting example, we therefore consider the STLC with
booleans (where the base type $\basety$ is the type of booleans).
Consider the $\Ty$-sorted first-order presentation $\BoolPres$ with
two $(\ec; \basety)$-ary operators $\oper{true}$, $\oper{false}$, and,
for each $\A \in \Ty$, a $(\basety, \A, \A; \A)$-ary operator
$\oper{ite}$ (``if-then-else''), along with two equations:
\begin{gather*}
  y : \A, z : \A \vdash
    \oper{ite}(\oper{true}(), y, z)
    \tmeq y : \A
  \qquad
  y : \A, z : \A \vdash
    \oper{ite}(\oper{false}(), y, z)
    \tmeq z : \A
\end{gather*}
Let $\BoolClone$ be the $\Ty$-sorted clone that is presented by
$\BoolPres$.

Consider the free $\STLCBetaEta$-algebra $F_{\STLCBetaEta} \BoolClone$,
and the $\STLCBetaEta$-algebra $\Model_\Bool$ (as defined in \Cref{example:stlc-algebras}) with $\Bool = \{\btt, \bff\}$.
The former should be thought of as containing the terms of the STLC with
booleans (we make this precise below); the latter is the usual model in
$\Set$.
Both have clone homomorphisms from $\BoolClone$: the free algebra has
$\eta_{\BoolClone} : \BoolClone \to F_{\STLCBetaEta} \BoolClone$; the model $\Model_\Bool$ has
the unique $g : \BoolClone \to \Model_\Bool$ such that
\begin{gather*}
  g_{\Gamma; \basety} (\oper{true}())
    = \zeta \mapsto \btt
  \qquad
  g_{\Gamma; \basety} (\oper{false}())
    = \zeta \mapsto \bff
  \\
  g_{\Gamma; \A} (\oper{ite}(t_1,t_2,t_3))
    = \zeta \mapsto \begin{cases}
      g_{\Gamma;\A}(t_2)(\zeta)
        &\text{if}~g_{\Gamma;\basety}(t_1)(\zeta) = \btt
      \\
      g_{\Gamma;\A}(t_3)(\zeta)
        &\text{if}~g_{\Gamma;\basety}(t_1)(\zeta) = \bff
    \end{cases}
\end{gather*}
The algebra homomorphism
$\extend g : F_{\STLCBetaEta}\BoolClone \to \Model_\Bool$ gives the
interpretation of STLC terms in the model.
Define a subset
$
  P(\basety)
    \subseteq (F_{\STLCBetaEta}\BoolClone \times \Model_\Bool)(\ec; \basety)
    = (F_{\STLCBetaEta}\BoolClone)(\ec; \basety) \times \Bool
$
by
\[
  P(\basety)
    = \{(\eta_{\BoolClone}(\oper{true}()),\btt)
      , (\eta_{\BoolClone}(\oper{false}()),\bff)\}
\]
This extends to a logical relation $P$ by the definition on
function types above and, by a simple proof, satisfies the precondition
of \cref{lemma:lr-closed}, where the clone homomorphism $f$ is
$
  \tuple{\eta_{\BoolClone}, g}
    : \BoolClone \to F_{\STLCBetaEta}\BoolClone \times \Model_\Bool
$.
Hence, for all $t \in (F_{\STLCBetaEta}\BoolClone)(\ec; A)$, we have
$
  (t, \extend g_{\ec; \A}(t))
    = \extend{\tuple{\eta_{\BoolClone}, g}}_{\ec; \A}(t)
    \in P(\A)
$.
When $\A = \basety$ this immediately implies, for all
$t, t' \in (F_{\STLCBetaEta}\BoolClone)(\ec; \basety)$, that if
$\extend g_{\ec; \basety}(t) = \extend g_{\ec; \basety}(t')$ then
$t = t'$.

This last property is seen to be adequacy of the set-theoretic model $\Model_\Bool$ as follows.
Let $\STLCBetaEtaBoolClone$ be the $\Ty$-sorted clone that is defined in
the same way as $\STLCBetaEtaClone$ (\cref{example:stlcclone})
but with additional term formers and equations (omitting the typing constraints on equations):
\[
  \begin{array}{l}
  \begin{prooftree}[center=false]
    \infer0{\Gamma \vdash \strue : \basety}
  \end{prooftree}
  \\[1.5ex]
  \begin{prooftree}[center=false]
    \infer0{\Gamma \vdash \sfalse : \basety}
  \end{prooftree}
  \end{array}
  \qquad
  \begin{prooftree}
    \hypo{\Gamma \vdash t_1 : \basety}
    \hypo{\Gamma \vdash t_2 : \A}
    \hypo{\Gamma \vdash t_3 : \A}
    \infer3{\Gamma \vdash \sif{t_1}{t_2}{t_3} : \A}
  \end{prooftree}
  \qquad
  \begin{array}{r@{~\tmeq_{\beta\eta}~}l}
    \sif{\mathrlap{\strue}\phantom{\sfalse}}{t_2}{t_3} & t_2 \\
    \sif{\sfalse}{t_2}{t_3} & t_3 \\
  \end{array}
\]
$\STLCBetaEtaBoolClone$ forms an $\STLCBetaEta$-algebra, and there is a
clone homomorphism $\eta : \BoolClone \to \STLCBetaEtaClone$ making it
into the free $\STLCBetaEta$-algebra on $\BoolClone$.
Hence we can apply the method above with
$F_{\STLCBetaEta}\BoolClone = \STLCBetaEtaBoolClone$.
The algebra homomorphism
$
  \extend g : \STLCBetaEtaBoolClone
    \to \Model_\Bool
$
sends each term $\Gamma \vdash t : \A$ to its interpretation as a
function $\prod_i \Model_\Bool\sem{\Gamma_i}\to \Model_\Bool{\sem{\A}}$.
Adequacy is therefore exactly the property that
$\extend g_{\ec; \basety}(t) = \extend g_{\ec; \basety}(t')$
implies $t = t'$.

\subsection{Open terms and normalization}
\label{section:open-terms-normalization}
As a second example, we show that every term of the STLC is equal (up to
$\beta\eta$-equality) to one in \emph{$\eta$-long $\beta$-normal form}
(we define these normal forms below).
The proof mostly follows Fiore~\cite{fiore2002semantic}, except that we reason abstractly
using the universal property of free algebras via our induction
principle.
It makes use of \emph{Kripke logical relations} (with varying arity), which
were introduced by Jung and Tiuryn~\cite{jung1993new} to study
$\lambda$-definability.

We first show that our induction principle enables arguments using
Kripke logical relations over the STLC.
Fix a $\STLCBetaEta$-algebra $(\Y, \sem{{-}})$, homomorphism
$f : \X \to \Y$ from a clone $\X$, and a subset
$P(\Gamma;\basety) \subseteq Y(\Gamma;\basety)$ for each $\Gamma$.
We extend $P$ from the base type $\basety$ to all types by
\[
  P(\Gamma;\mspace{-1mu}A{\To}B)
    = \{\mspace{-1mu}t\mspace{2mu}{\in}\mspace{2mu} Y(\Gamma;\mspace{-2mu}A{\To}B)
    \mid
    \forall \Delta,\mspace{-2mu}
    \brho\mspace{2mu}{\in}\mspace{2mu}\Var S(\Delta;\mspace{-3mu}\Gamma),
    a\mspace{2mu}{\in}\mspace{2mu}P(\Delta;\mspace{-2mu}A).~
      \sem{\oper{app}}_{\Delta} (\subst{t}{\rename \brho}, a)
        \mspace{3mu}{\in}\mspace{3mu} P(\Delta;\mspace{-3mu}B)
  \mspace{-1mu}\}
\]
This is the standard definition of a Kripke logical relation on function
types (other than using all renamings $\brho$ rather than just
weakenings, which is inessential).
We therefore have a family of subsets
$P(\Gamma; A) \subseteq Y(\Gamma; A)$,
to which we apply \cref{corollary:open-terms} and obtain the
following.
\begin{lemma}\label{lemma:lr-open}
  If the family of subsets $P$ satisfies
  \begin{itemize}
    \item for every context $\Gamma$ we have
      $\var[\Gamma] \in P(\Gamma; \Gamma)$;
    \item for every variable renaming $\brho \in \Var S(\Delta; \Gamma)$
      and term $t \in P(\Gamma; \basety)$ we have
      $\subst{t}{\rename \brho} \in P(\Delta; \basety)$;
    \item for every term $t \in X(\Gamma; B)$ and
      substitution $\boldsymbol\sigma \in P(\Delta; \Gamma)$ we have
      $\subst{(f_{\Gamma; B} t)}{\boldsymbol\sigma} \in P(\Delta; B)$,
  \end{itemize}
  then, for every free term $t \in (F_{\STLCBetaEta} \X)(\Gamma; A)$, we have
  $\extend f_{\Gamma; A}(t) \in P(\Gamma; A)$.
\end{lemma}
\begin{proof}
  The only non-trivial assumption of \cref{corollary:open-terms} is
  closure under operators.
  For closure under $\oper{app}$, if
  $t \in \osubstpred P(\Gamma; A \To B)$ and
  $u \in \osubstpred P(\Gamma; A)$, then, for all
  $\bsigma \in P(\Delta;\Gamma)$, we have
  $
    \subst{(\sem{\oper{app}}_\Gamma(t,u))}{\bsigma}
    = \sem{\oper{app}}_\Delta(\subst t \bsigma, \subst u \bsigma)
  $,
  because interpretations of operators commute with substitution; this
  is an element of $P(\Delta; B)$ using $\subst t \bsigma \in P(\Delta; A\To B)$
  on the identity variable-renaming.
  For closure under $\oper{abs}$, suppose that
  $t \in \osubstpred P(\Gamma, A; B)$.
  The assumption of the present lemma that $P$ is closed under variable
  renamings at the base type $\basety$ extends to all types $A$ by an
  easy induction on $A$.
  For every $\bsigma \in P(\Delta; \Gamma)$,
  $\brho \in \Var S(\Xi; \Delta)$, and $a \in P(\Xi; A)$, we then have that
  $\subst t {(\bsigma \compose \rename \brho), a} \in P(\Xi; B)$.
  Preservation of substitution by $\sem{\oper{abs}}$, and the
  \ref{eq:algebra-beta} law, together imply that
  $
    \sem{\oper{app}}_{\Xi}(
      \subst{\subst{(\sem{\oper{abs}}_{\Gamma}(t))}{\bsigma}}
        {\rename \brho})
    =
    \subst t {(\bsigma \compose \rename \brho), a} \in P(\Xi; B)
  $.
  Hence $\sem{\oper{abs}}_{\Gamma}(t) \in \osubstpred P(\Gamma; A\To B)$
  as required.
\end{proof}

We use this to show normalization as follows.
\emph{Normal forms} $\Gamma \vdash_n t : A$ are defined mutually
inductively with the \emph{neutral forms} $\Gamma \vdash_m t : A$ by
the following rules:
\[
  \begin{prooftree}[center=false]
    \infer0{\Gamma, x : A, \Delta \vdash_m x : A}
  \end{prooftree}
  \quad
  \begin{prooftree}[center=false]
    \hypo{\Gamma \vdash_m f : A \To B}
    \hypo{\Gamma \vdash_n a : A}
    \infer2{\Gamma \vdash_m \sapp\,f\,a : B}
  \end{prooftree}
  \quad
  \begin{prooftree}[center=false]
    \hypo{\Gamma \vdash_m t : \basety}
    \infer1{\Gamma \vdash_n t : \basety}
  \end{prooftree}
  \quad
  \begin{prooftree}[center=false]
    \hypo{\Gamma, x : A \vdash_n t : B}
    \infer1{\Gamma \vdash_n \lambda x : A.\,t : A \To B}
  \end{prooftree}
\]
Consider the initial $\STLCBetaEta$-algebra, which is the clone
$\STLCBetaEtaClone$ of STLC terms up to $\tmeq_{\beta\eta}$.
We write $\Nf(\Gamma; A)$ for the subset of STLC terms that are equivalent to a term in normal form under $\tmeq_{\beta\eta}$; and likewise write
$\Ne(\Gamma; A)$ for neutral forms.
We consider both as subsets of $\STLCBetaEtaTerms(\Gamma; A)$; both are
closed under variable renaming.
The family of subsets we consider,
$P(\Gamma; A) \subseteq \STLCBetaEtaTerms(\Gamma; A)$,
is defined on the base type as
$P(\Gamma; \basety) = \Nf(\Gamma; \basety)$, and on other types by the
logical relations definition above.
By a simple induction on the sort $A$, one can show that
$\Ne(\Gamma; A) \subseteq P(\Gamma; A) \subseteq \Nf(\Gamma; A)$
(e.g.~as in~\cite{fiore2002semantic}).
Since variables are neutral, this tells us in particular that
$\var[\Gamma] \in P(\Gamma; \Gamma)$ for all $\Gamma$.
It then follows from \cref{lemma:lr-open} that
$t = \extend{(\rename)}(t) \in P(\Gamma; A) \subseteq \Nf(\Gamma; A)$
for all $t \in \STLCBetaEtaTerms(\Gamma; A)$, and so that every term of the
STLC is $\beta\eta$-equivalent to one in normal form.

\section{Comparison to other approaches}
\label{comparison}

While we promote abstract clones as an elementary approach to simple
type theories (qua multisorted second-order abstract syntax), there are
several equivalent concepts that have been used to similar effect.
We give a brief overview of the existing literature on the subject and
a comparison with our work; we give references where possible, but unfortunately some of the relationships here exist only in the mathematical folklore.

\paragraph*{Presheaves and substitution monoids}

The study of second-order abstract syntax was initiated by Fiore et al.\ \cite{fiore1999abstract, fiore2002semantic},
who represent term structure using presheaf categories.
In their setting, one considers functors $T : \mathbb L(S)^{\mathrm{op}} \to \Set^S$, where $\mathbb L(S)$ is the category in which objects are contexts $\Gamma$, and morphisms $\brho : \Delta \to \Gamma$ are variable renamings $\brho \in \Var S (\Gamma; \Delta)$ (recall \Cref{subst-and-renaming}).
The $S$-indexed sets $T(\Gamma)$ consist of the sorted terms in context $\Gamma$; while
the functions $T(\brho)$ rename the variables inside the terms to change their context.
Substitution is accounted for by considering the monoidal structure
$(\bullet, V)$ on $[\mathbb L(S)^{\mathrm{op}}, \Set^S]$, in which $T \bullet T'$
represents (for each context $\Gamma$) the simultaneous substitution of
each variable in $T$ with a term from $T'$, and $V$ represents the variables in each context.
Monoids with respect to this structure are equipped with variables and substitution
operations; they are equivalently abstract clones~\cite[Proposition~3.4]{fiore1999abstract}.
Fiore and Hur~\cite{fiore2010second} define $\SIGMA$-algebras as monoids
in $[\mathbb L(S)^{\mathrm{op}}, \Set^S]$ equipped with interpretations of
the operators of a presentation $\SIGMA$ satisfying its equations; they are equivalent to our
$\SIGMA$-algebras.
Our setting is therefore equivalent to that of Fiore et al. The advantage of our approach is that abstract clones require less categorical machinery; for those comfortable with category theory, this will be less of a concern.

There are some technical differences with previous work.
Fiore and Hur~\cite{fiore2010second} show the existence of the free
$\SIGMA$-algebras on each presheaf $T$; in light of our free algebra
result, the construction of the free algebra on $T$ can be factored into two steps:
constructing the free clone $\X$ on $T$ by freely adding variables and
substitution, and then taking the free $\SIGMA$-algebra on the clone
$\X$.
In our examples above, we begin with a clone that admits substitution,
and hence do not freely add substitution.
In a separate treatment, Hofmann~\cite{hofmann1999semantical} gives an induction principle for
the $\lambda$-calculus using presheaves, but only considers predicates over
closed terms; we obtain induction for closed terms as a corollary of
induction over open~terms.

\paragraph*{Cartesian multicategories}

Each abstract clone $\X$ has an identity operation for every sort $\B$, given by the unique variable projection $\var[[\B]]_1 \in X(\B; \B)$, along with admissible operations of exchange, weakening, and contraction. In this way, the sets of terms $X(\Gamma; \A)$ form the structure of a \emph{cartesian multicategory} with object set $S$ (intuitively a category whose morphisms may have multiple inputs, subject to the structural properties of first-order equational logic).
Conversely, every cartesian multicategory gives rise to an abstract clone. Thus, one could carry out the development of this paper in the context of cartesian multicategories (cf.\ \cite[Section~9]{arkor2020algebraic}). Clones are our preferred choice, because the definition of clone (in which projections are the primary operation) provides a more minimal axiomatisation than that of cartesian multicategory (in which the structural operations are primary).
Note that one-object cartesian multicategories are usually called \emph{cartesian operads}, which correspond to monosorted abstract clones.

\paragraph*{Algebraic theories}

The traditional approach to describing first-order algebraic structure in categorical logic is through \emph{algebraic theories} \cite{lawvere1963functorial}. An algebraic theory is represented by a category with cartesian products, which permit the multimorphisms of a cartesian multicategory to be represented by morphisms from a product: for a context $[\A_1, \ldots, \A_n]$, the terms $x_1 : \A_1, \ldots, x_n : \A_n \vdash t : \B$ are represented by a hom-set $X(\A_1 \times \cdots \times \A_n, \B)$. The relationship between cartesian multicategories and algebraic theories is the notion of \emph{representability} for cartesian multicategories \cite{pisani2014sequential}.
Second-order structure in the context of algebraic theories is captured by \emph{second-order algebraic theories} \cite{fiore2010algebraic, mahmoud2011second, arkor2020higher}, which generalize the first-order setting by introducing exponential objects that represent function types. Every second-order presentation $\SIGMA$ induces a second-order algebraic theory, the algebras for which are given by taking coslices over $\SIGMA$~\cite{arkor2020higher}.

\paragraph*{Monads and relative monads}

There is a classical correspondence in category theory between algebraic theories and certain monads on the category of sets \cite{linton1969outline}, which in turn are equivalent to $J$-relative monads, for $J$ the inclusion of finite sets into sets \cite{altenkirch2010relative}. This has led to a line of investigation in which monads are used directly for second-order abstract syntax \cite{hirschowitz2007modules, hirschowitz2010modules, ahrens2016modules, ahrens2019modular, hirschowitz2020modules}. There are strong connections between this approach and that of presheaves and substitution monoids: for a detailed comparison, see the thesis of Zsid\'o~\cite{zsido2010typed}. In particular, the distinction between abstract clones and $J$-relative monads is slight, and the results of our development could equivalently be rephrased as statements about relative monads (cf.~\cite{arkor2020higher}).

\section{Conclusion}

We have shown that the abstract syntax of simple type theories has an
elementary treatment using abstract clones.
The framework we describe allows the specification of the terms and
equations of type theories via second-order
presentations~\cite{fiore2010second,fiore2010algebraic}.
Free algebras then give the syntax along with an accompanying induction principle, which we show
enables abstract proofs of non-trivial properties such as adequacy.
We emphasize that abstract clones axiomatize the syntax only of
\emph{simple} type theories:
clones cannot express linear types, dependent types, or type theories
in which variables stand only for certain classes of term
(e.g.~polarized type theories~\cite{zeilberger2009logical}, and
the call-by-value $\lambda$-calculus). In some cases, analogous structures are already known (for instance, symmetric multicategories for linear type theories~\cite{tanaka2000abstract, huot2016operads}); for others, such as dependent type theories, this remains an open problem.

\bibliography{fscd21}

\appendix
\section{Normalization with global state}
\label{normalization-global-state}

As a further example of the application of abstract clones to problems motivated by simple type theories, we prove a normalization result for the STLC with $V$-valued global state: concretely, this calculus is given by the free algebra of the second-order
presentation $\STLCBetaEta$ of the STLC with $\beta\eta$-equality
on the clone $\GSClone V$ of $V$-valued global state,
whose syntax is described in \cref{example:stlc-initial}.
The proof is similar to normalization of the STLC without global state
(\cref{section:open-terms-normalization}); in particular, we reuse
\cref{lemma:lr-open}.

Recall that for $V = \{v_1, \dots, v_k\}$, the free algebra consists of
the syntax of the STLC extended by the additional term formers $\sget$ and
$\sput_{v_i}$.
Normal and neutral forms are defined as in
\cref{section:open-terms-normalization}, except with
\[
  \text{the rule~~}
  \begin{prooftree}
    \hypo{\Gamma \vdash_m t : \basety}
    \infer1{\Gamma \vdash_n t : \basety}
  \end{prooftree}
  \text{~~replaced by~~}
  \begin{prooftree}
    \hypo{\Gamma \vdash_m t_1 : \basety}
    \hypo{\cdots}
    \hypo{\Gamma \vdash_m t_k : \basety}
    \infer3[\textnormal{($w_1, \dots, w_k \in V$)}]
      {\Gamma \vdash_n
        \sget(\sput_{w_1}(t_1), \dots, \sput_{w_k}(t_k)) : \basety}
  \end{prooftree}
  \text{.}
\]
Again we write $\Nf(\Gamma; A)$ (respectively $\Ne(\Gamma; A)$) for the
subsets of terms equal to a normal (respectively neutral) form, and
define the logical relation $P(\Gamma; A)$ on the base type as
$P(\Gamma; \basety) = \Nf(\Gamma; \basety)$, and on other types by the
logical relations definition in \cref{section:open-terms-normalization}.
Again we have
$\Ne(\Gamma; A) \subseteq P(\Gamma; A) \subseteq \Nf(\Gamma; A)$ by
induction on $A$; the only difference with the previous proof is that on
base types one has $\Ne(\Gamma; \basety) \subseteq \Nf(\Gamma; \basety)$,
because for $t \in \Ne(\Gamma; \basety)$ we have
$
  t \tmeq_{\beta\eta} \sget(\sput_{v_1}(t), \dots, \sput_{v_k}(t))
  \in \Nf(\Gamma; \basety)
$.
To prove that every term is equal to one in normal form up to
$\tmeq_{\beta\eta}$, it suffices to apply \cref{lemma:lr-open} with $f$
the clone homomorphism
$\eta_{\GSClone V} : \GSClone V \to F_{\STLCBetaEta}{\GSClone V}$.
The first two assumptions of the lemma have the same proof as before.
For the third, since $\GSClone V$ is presented by $\GSPres V$ and clone
homomorphisms preserve variables and substitution, it suffices to show
that
\begin{itemize}
  \item for each $t_1, \dots, t_k \in P(\Gamma; \basety)$, we have
    $\sget(t_1, \dots, t_k) \in P(\Gamma; \basety)$;
  \item for each $t \in P(\Gamma; \basety)$ and $i \le k$, we have
    $\sput_{v_i}(t) \in P(\Gamma; \basety)$.
\end{itemize}
The first statement holds because if
$
  t_i =
  \sget(\sput_{w_{i1}}(t'_{i1}), \dots, \sput_{w_{ik}}(t'_{ik}))
$ then
\[
  \sget(t_1,  \dots, t_k)
  \tmeq_{\beta\eta}
  \sget(\sput_{w_{11}}(t'_{11}), \dots, \sput_{w_{kk}}(t'_{kk}))
  \in \Nf(\Gamma; \basety) = P(\Gamma; \basety)
\]
The second statement holds because if
$
  t =
  \sget(\sput_{w_1}(t'_1), \dots, \sput_{w_k}(t'_k))
$ then
\[
  \sput_{v_i}(t)
  ~\tmeq_{\beta\eta}~
  \sput_{w_i}(t'_i)
  ~\tmeq_{\beta\eta}~
  \sget(\sput_{w_i}(t'_i), \dots, \sput_{w_i}(t'_i))
  \in \Nf(\Gamma; \basety) = P(\Gamma; \basety)
\]

\end{document}